%% file: main.tex
\documentclass[sigconf, nonacm, twocolumn]{acmart}
\newcommand\vldbauthors{\authors}
\newcommand\vldbtitle{\shorttitle} 
\newcommand\vldbpagestyle{plain} 
\renewcommand\footnotetextcopyrightpermission[1]{}
\usepackage{tabularx}
\usepackage{array}
\usepackage{booktabs}
\usepackage{amsmath} 
\usepackage{enumitem}
\usepackage{subcaption}
\usepackage{multirow}   
\usepackage{xcolor}
\usepackage{balance}
\usepackage{mathtools}
\usepackage{pifont}
\usepackage{setspace}
\usepackage[linesnumbered,ruled,vlined]{algorithm2e}
\newcommand{\revision}[1]{\textcolor{black}{#1}}
\newcommand{\yin}[1]{\textcolor{black}{#1}}
\newcommand{\remove}[1]{\textcolor{black}{#1}}
\usepackage{pdfpages}

\newcommand{\bbc}{{BBC}\xspace}
\newtheorem{observation}{Observation}

\newcommand\vldbdoi{XX.XX/XXX.XX}
\newcommand\vldbpages{XXX-XXX}
\newcommand\vldbvolume{14}
\newcommand\vldbissue{1}
\newcommand\vldbyear{2026}

\begin{document}

\title{BBC: Improving Large-$k$ Approximate Nearest Neighbor Search with a Bucket-based Result Collector}

\author{Ziqi Yin$^1$, Gao Cong$^1$, Kai Zeng$^2$, Jinwei Zhu$^2$, Bin Cui$^3$}
\affiliation{
\begin{center}
    {$^1$Nanyang Technological University, Singapore} \quad
    {$^2$Huawei Technologies Co., Ltd} \quad
    {$^3$Peking University, China} ~~~~~
\end{center}
}
\email{ziqi003@e.ntu.edu.sg,  gaocong@ntu.edu.sg, {zhujinwei, kai.zeng}@huawei.com, bin.cui@pku.edu.cn}
\input{abstract}
\maketitle

\vspace*{-1em}
\pagestyle{\vldbpagestyle}
\begingroup\small\noindent\raggedright\textbf{PVLDB Reference Format:}\\
\revision{\vldbauthors. \vldbtitle. PVLDB, \vldbvolume(\vldbissue): \vldbpages, \vldbyear.}\\
\revision{\href{https://doi.org/\vldbdoi}{doi:\vldbdoi}}
\endgroup
\begingroup

\vspace*{-1em}
\renewcommand\thefootnote{}\footnote{\noindent
\revision{This work is licensed under the Creative Commons BY-NC-ND 4.0 International License. Visit \url{https://creativecommons.org/licenses/by-nc-nd/4.0/} to view a copy of this license. For any use beyond those covered by this license, obtain permission by emailing \href{mailto:info@vldb.org}{info@vldb.org}. Copyright is held by the owner/author(s). Publication rights licensed to the VLDB Endowment. \\
\raggedright Proceedings of the VLDB Endowment, Vol. \vldbvolume, No. \vldbissue\ %
ISSN 2150-8097. \\
\href{https://doi.org/\vldbdoi}{doi:\vldbdoi} \\}
}\addtocounter{footnote}{-1}\endgroup

\vspace*{-1em}
\ifdefempty{\vldbavailabilityurl}{}{
\vspace{.3cm}
\begingroup\small\noindent\raggedright\textbf{PVLDB Artifact Availability:}\\
\revision{The source code, data, and/or other artifacts have been made available at \url{https://github.com/Heisenberg-Yin/BBC}.}
\endgroup
}

\input{intro}

\input{problem}
\input{framework}

\input{exp}

\input{related}

\input{conclusion}

\newpage
\bibliographystyle{ACM-Reference-Format}
\bibliography{ref}
\input{appendix}

\end{document}

%% file: abstract.tex
\begin{abstract}
Although Approximate 
Nearest Neighbor (ANN) search has been extensively studied, large-$k$ 
ANN queries that aim to retrieve a large number of nearest neighbors remain underexplored, despite their numerous real-world applications. Existing 
ANN methods face significant performance degradation for such queries. In this work, we first investigate the reasons for the performance degradation of quantization-based ANN indexes: (1) the inefficiency of existing top-$k$ collectors, which incurs significant overhead in candidate maintenance, and (2) the reduced pruning effectiveness of quantization methods, which leads to a costly re-ranking process.
To address this, we propose a novel bucket-based result collector {(\bbc)} 
to enhance the efficiency of existing quantization-based ANN indexes for large-$k$ ANN queries. {\bbc} introduces two key components: (1) a bucket-based 
result buffer that 
organizes candidates into buckets by their distances to the query.
This design reduces ranking costs and improves cache efficiency,
enabling high-performance maintenance of a candidate superset and a lightweight final selection of top-$k$ results. 
%
(2) two re-ranking algorithms tailored for different types of quantization methods, which accelerate their re-ranking process by reducing either the number of candidate objects to be re-ranked or cache misses. Extensive experiments on real-world datasets demonstrate that {\bbc} accelerates existing quantization-based ANN methods by up to 3.8$\times$ at recall@$k$ = 0.95 for large-$k$ ANN queries.
\end{abstract}      


%% file: intro.tex
\section{Introduction}
\label{sec:intro}



Driven by {the rapid process of} large-scale machine learning and generative AI techniques, efficient vector search has become a critical capability in modern data systems~\cite{patel2025semantic, yang2020pase, mohoney2023high,pound2025micronn}. 
Vector databases~\cite{wang2024vector, milvus2021, pan2024vector, DBLP:journals/vldb/PanWL24} now serve as the foundation for querying embeddings 
generated by deep {learning} models, where Approximate 
Nearest Neighbor (ANN) search is the core computational primitive~\cite{wang2021comprehensive, DBLP:journals/vldb/PanWL24}, trading off minor accuracy for significantly improved efficiency~\cite{curse1,curse2}. 
In practice, ANN algorithms are typically extended to 
retrieve approximate $k$-nearest neighbors to meet the demands of real-world applications.

Although ANN queries have been extensively studied, most existing studies design and evaluate their methods under small $k$ settings, 
which are typically in the range of a few tens to a few  hundreds~\cite{liApproximateNearestNeighbor2020,dobson2023scaling}.
{This setting is well-suited for some applications such as retrieval-augmented generation (RAG)~\cite{lewis2020retrieval}, 
where an ANN index retrieves the top-10 relevant documents for a large language model to generate the final {responses}. 
However, many real-world applications 
involve large $k$ scenarios (e.g., $k \geq 5{,}000$), where a large number of nearest neighbors need to be retrieved for each query. 
We refer to such queries as large-$k$ ANN queries and next present several of their applications.
%
\begin{enumerate}[leftmargin=*,topsep=0pt]
    \item 
    In model training or fine-tuning scenarios, it is often necessary to efficiently retrieve a large set of highly relevant data samples to construct the training dataset. 
These samples typically number in the tens of thousands in many real-world applications, such as retrieving videos or images that capture specific types of 
dangerous driving behavior. In such cases, 
the initial query is often ambiguous, such as an image representing a driving behavior. 
Data engineers usually need to perform multiple iterations of search, refining the query by selecting better examples from the retrieved results 
before identifying an effective query and obtaining a satisfactory set of results.
    \item 
    In document retrieval, state-of-the-art methods often adopts a retrieve-and-rerank pipeline~\cite{lee2023rethinking}. Documents are encoded into embeddings using pre-trained language models~\cite{ANNForDPR}. An ANN index built on these embeddings retrieves tens of thousands of candidate documents for each query. Subsequently, a more sophisticated model, such as ColBERT~\cite{khattab2020colbert}, 
    which encodes queries and documents into token embeddings and computes similarity by aggregating token-level similarities, re-ranks the candidates to obtain the final results.

    \item 
    In industrial 
recommendation systems~\cite{DBLP:conf/www/KhandagaleJASYW25,gao2021learning}, 
hundreds of thousands of candidates are first retrieved via an ANN index and then re-ranked using more computationally expensive 
models to produce the final recommendations. 
\end{enumerate}
However, existing ANN methods face significant performance degradation when handling large-$k$ ANN queries, as 
demonstrated in our evaluation on four representative ANN indexes: the IVF~\cite{johnsonBillionscaleSimilaritySearch2019}, the popular graph-based method HNSW~\cite{malkovEfficientRobustApproximate2020}, and two quantization-based methods IVF+RaBitQ~\cite{gao2024rabitq} and IVF+PQ~\cite{jegouProductQuantizationNearest2011}. 
An example result on the {C4} dataset is shown in Figure~\ref{fig:qs-exp-motivation} and similar trends are observed across other datasets. 
We observe that at recall$@k$ = 0.95, when $k$ increases from 100 to 5,000, the throughput of IVF+RaBitQ
drops from 227 queries per second (QPS) 
to 47 QPS, 
a 4.8 $\times$ slowdown; HNSW's throughput falls from 113 QPS 
to 20 QPS, 
showing a larger
5.7$\times$ slowdown. 
In this work, we focus on optimizing quantization-based methods for large-$k$ ANN queries for two reasons: 1) quantization-based methods exhibit superior performance 
for large-$k$ ANN queries; and 2) our empirical results and  analysis\footnote{The suboptimal performance of HNSW at large $k$ arises from the fact that graph-based ANN indexes are designed for small $k$. These methods construct a proximity graph during indexing and use it to navigate queries toward nearby objects to reduce search space. However, when $k$ is large, the graph traversal inevitably expands to a larger portion of the graph, incurring significant additional overhead, as shown in Figure~\ref{fig:qs-exp-time-portion}.} show that quantization-based methods are more robust to increase of $k$ compared with graph-based methods such as HNSW. 

\begin{figure}[!t]
\centering
\includegraphics[width=0.75\columnwidth]{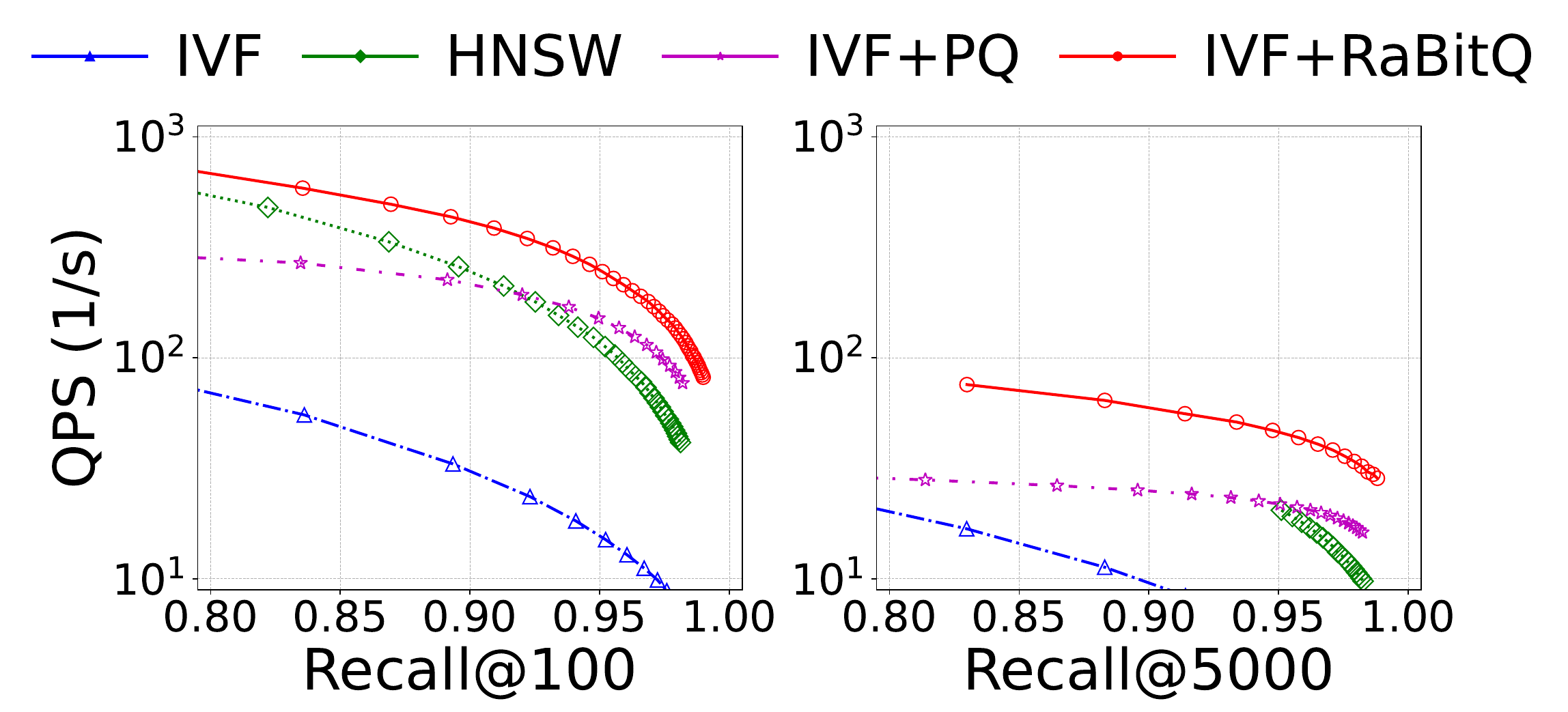}
\vspace*{-0.5em}
\caption{Querying Performance of {IVF}, {HNSW}, {IVF+PQ}, and {IVF+RaBitQ} on the C4 dataset at $k$ = 100 and $k$ = 5000.}
\vspace*{-1.5em}
\label{fig:qs-exp-motivation}
\end{figure}

\noindent
{\bf Challenge 1.}
Quantization-based methods often face significant slowdowns in large-$k$ ANN queries due to two primary challenges. 
The first stems from the inefficiency of existing top-$k$ collectors when handling large $k$.
%
These collectors are responsible for maintaining the $k$ nearest candidates by storing each candidate's ID and distance to the query. {Whenever a closer candidate is found, it replaces the farthest one, whose distance serves as the comparison threshold.} 
Existing ANN studies commonly employ a binary-heap priority queue as the top-$k$ collector. However, when $k$ is large, the heap becomes inefficient under modern memory hierarchies consisting of L1, L2, L3 caches and main memory. 
The inefficiency arises because 
the heap size exceeds the capacity of the fastest L1 cache. 
For example, when $k$ = 100, the distance–ID pairs 
occupy only 
800 bytes, which is well within the 32 KB L1 cache capacity, allowing the heap reside entirely in the L1 cache 
and achieve low latency. However, when $k$ = 5,000, the heap grows to $40{,}000$ bytes, exceeding the L1 cache capability and causing frequent L1 cache misses that significantly increase latency. 
For example, in {IVF+RaBitQ}, its share of runtime rises from 2\% at $k$ = 100 to 23\% at $k$ = 5{,}000 under recall@$k$ = 0.95, as shown in Figure~\ref{fig:qs-exp-time-portion}. This observation is consistent with prior findings~\cite{larkin2014back}
that link priority queue performance to L1 cache misses.
\noindent
{\bf Challenge 2.}
{As $k$ increases, the pruning effectiveness 
of quantization methods drops.
These methods accelerate ANN search by estimating distances to quickly prune distant objects and can be grouped into two categories: 
(1) 
\textit{unbounded} methods that prune solely by estimated distances (e.g., PQ),
and (2) {\textit{bounded}} methods that provide 
{probabilistically guaranteed distance bounds} 
(e.g., RaBitQ).
Although they differ in querying strategies, both types of approaches require re-ranking a growing number of candidates as $k$ increases.
In \textit{bounded} methods, the candidate set is maintained by a top-$k$ collector, and any object whose estimated bounds overlap with the current threshold 
is re-ranked. As $k$ grows, the number of such objects increases, leading to higher re-ranking overhead.
For example, in {IVF+RaBitQ}, the runtime share of exact distance computation rises from 27\% at $k$ = 100 to 40\% at k = 5,000, as shown in Figure~\ref{fig:qs-exp-time-portion}. 
%
Similarly, {\textit{unbounded}} methods retrieve and re-rank a candidate set whose size is typically several times larger than $k$, to maintain high recall. 
%
Consequently, as $k$ increases, the re-ranking cost grows roughly linearly, resulting in a significant slowdown.

\begin{figure}[!t]
\centering
\includegraphics[width=0.75\columnwidth]{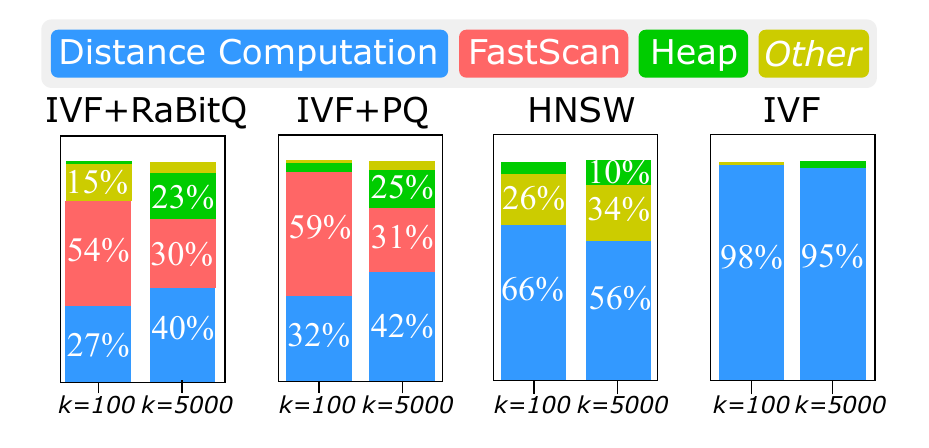}
\vspace*{-1em}
\caption{Time Overhead Breakdown of four methods at different $k$, where ``Distance computation'' denotes exact distance computation, ``FastScan'' denotes estimated distance computation, ``Heap'' denotes heap operations, and ``Other'' covers the remaining costs.}
\vspace*{-1em}
\label{fig:qs-exp-time-portion}
\end{figure}

\noindent\textbf{Our Method.} To address the first challenge, 
we observe that existing collectors typically maintain an exact top-$k$ set, where each stored candidate may be accessed and 
replaced. When $k$ is large and the stored distance-id pairs exceed the L1 cache capacity, this 
results in frequent L1 cache misses and high maintenance overhead. 
To overcome this, we propose a novel \textbf{bucket-based result buffer} ({\bbc}) that reorganizes candidate storage to maximize cache efficiency, {preserving exact results without maintaining exact top-$k$ order.} 
%
Specifically, it partitions the estimated 
distance range between the query and data objects into non-overlapping sub-ranges through one-dimensional quantization. Each sub-range 
corresponds to a bucket consisting of two linear buffers that sequentially store the IDs and distances of candidates falling within this sub-range. This design offers two key benefits: (1) 
it {lowers} 
ranking complexity by organizing candidates 
{into} buckets based on distance, {maintaining ordering across buckets while avoiding ordering within each bucket;} 
(2) it reduces L1 cache misses, as the sequential insertion pattern within each bucket enables the hardware prefetching mechanism to proactively prefetch relevant memory blocks into the L1 cache for subsequent candidate insertions.
Leveraging both bucket-level ranking and the number of objects stored in each bucket, we efficiently identify the buckets that contain the exact top-$k$ results and the threshold bucket that holds boundary candidates near the threshold distance. Together, these buckets form a \textit{candidate superset}. The upper bound distance of the threshold bucket serves as a \textit{relaxed threshold}. Both the candidate superset and 
the relaxed threshold can be maintained at low cost, since only a small number of buckets are involved. 
This design eliminates per-object access and replacement by operating only on bucket-level structures, where distant buckets are implicitly pruned once they become irrelevant. 
When the exact top-$k$ set is required, the final selection is restricted to the threshold bucket, as the bucket-level ranking guarantees that all preceding buckets contain only closer candidates. 
Leveraging the distance concentration phenomenon~\cite{curse1,curse2}, we further show that under an equal-depth partition of the distance range, {the error between the relaxed and exact thresholds are on the order of $10^{-2}$, keeping the selection cost negligible}, as 
supported by both theoretical guidance (Section~\ref{sec:result-buffer}) and empirical evaluation (Section~\ref{sec:qs-exp-results}).

{To address the second limitation, we design two new re-ranking algorithms tailored to different 
quantization methods.
For {\textit{bounded}} methods, 
we aim to skip re-ranking objects that are guaranteed to be either within  or outside the top-$k$ based on their estimated distance bounds, and  re-rank only uncertain ones. 
We first formulate an 
minimal re-ranking scenario that 
minimizes the number of re-ranked objects without sacrificing accuracy
and design a solution to achieve it. However, this approach incurs substantial heap overhead
that offsets its benefits. To address this, we design a greedy re-ranking algorithm that integrates seamlessly with the proposed result buffer, significantly reducing the number of re-ranked objects with a small extra cost. For \textit{unbounded} methods, the number of re-ranked objects cannot be reduced since their estimated distances 
lacking guaranties. Instead, we propose an early re-ranking strategy that tightly 
couples re-ranking with the result buffer. It computes exact distances for objects predicted to enter the re-ranking pool when their data are accessed, thus effectively reducing cache misses during exact distance computation.

{Building on these techniques, we develop a novel bucket-based result collector ({\bbc}) that substantially enhances the efficiency of existing quantization-based methods for large-$k$ ANN queries without compromising accuracy. {\bbc} integrates the proposed result buffer to gather candidates efficiently and incorporates the  newly designed re-ranking algorithms to produce the final results. 

The main contributions of this work are summarized as follows:

\begin{enumerate}[leftmargin=*,topsep=0pt]
\item We identify and analyze 
two major limitations of 
quantization-based 
methods for large-$k$ ANN queries: inefficiency of top-$k$ collectors and declined pruning effectiveness, both of which cause substantial performance degradation.
%
To our knowledge, these findings have not been reported in prior work. 

\item We propose a novel bucket-based result collector ({\bbc}), which introduces a 
bucket-based result buffer serving as a cache-efficient top-$k$ collector
and two new re-ranking algorithms tailored to {\textit{bounded}} and \textit{unbounded} quantization methods. 
{To the best of our knowledge, this is the first framework explicitly designed for large-$k$ ANN queries.}

\item Extensive experiments on real-world datasets show that:
    (1) {\bbc} accelerates existing quantization-based methods on large-$k$ ANN queries by up to 
    3.8$\times$ speedup at recall@$k$ = 0.95; 
    (2) the proposed result buffer reduces the overhead of the top-$k$ collector
    by an order of magnitude; 
    and (3) the new re-ranking algorithms speed up the re-ranking 
    efficiency by up to 1.8$\times$.

\end{enumerate}

%% file: problem.tex
\section{Background and Motivations}
\label{sec:preliminary-motivations}

\noindent\textbf{Problem Statement.}
Given a dataset $D$ of $N$ data objects, each represented by a $d$-dimensional vector,
{the Approximate Nearest Neighbor (ANN) query involves two phases: 
indexing and querying. In the indexing phase, it constructs a data structure based on the data vectors. In the querying phase, 
given a query $q$, it}
aims to efficiently retrieve the $k$ nearest vectors {
using the data structure \revision{under a similarity metric.}
}
Most existing studies~\cite{DBLP:journals/is/MalkovPLK14,jegouProductQuantizationNearest2011,jayaram2019diskann,milvus2021} {focus on the setup where} $k$ is small (e.g., $k=100$). 
{However, as discussed in Section~\ref{sec:intro}, the query with large $k$ (e.g., $k\ge5,000$) is important in many applications 
and introduces new challenges in algorithm design.}
In this study, we aim to develop an efficient solution for large-$k$ ANN queries \revision{($k \geq 5{,}000$) under commonly used similarity metrics in vector spaces, such as metric distances like Euclidean distance, as well as cosine similarity and inner product. Unless otherwise specified, we adopt Euclidean distance as the default metric and also discuss extensions to other similarity measures. }

\noindent\textbf{Modern Memory Hierarchy.} The modern memory hierarchy typically consists of L1, L2, and L3 caches and main memory. 
The L1 cache, located closest to the CPU core, provides the fastest access speed but has the smallest capacity, typically 32 KB. It stores the most frequently accessed data and instructions to minimize access latency. The L2 and L3 caches are located farther from the CPU cores, offering slower access than L1 but faster than main memory. Although main memory is much larger, typically ranging from several tens to hundreds of gigabytes (GB), its high access latency makes it slower. Therefore, the data required by the CPU are loaded into the L1 cache before processing. The transfer of data from main memory or from the L3/L2 caches to the L1 cache results in L1 cache misses, which introduce high latency, as shown in prior experimental evaluations~\cite{larkin2014back}. To reduce L1 cache misses, modern hardware automatically prefetches memory blocks adjacent to the currently accessed data into the L1 cache, while evicting less frequently used blocks to lower cache levels or main memory.


\noindent\textbf{Top-$k$ Collector.} Most ANN studies employ a {binary-heap priority queue} 
to maintain the $k$ nearest neighbors. 
Despite the binary heap achieving low latency when $k$ is small, it incurs frequent L1 cache misses and significantly higher latency {at larger values of $k$}
, as discussed in Section~\ref{sec:intro}. Although modern hardware supports L1 cache prefetching, it is only effective for data structures with regular memory access patterns, such as sequentially stored linear buffers. For data structures with irregular and unpredictable access patterns, such as binary heaps, its applicability is much more limited~\cite{malhotra2006library}.



\noindent\textbf{Quantization.} Quantization methods accelerate ANN search by efficiently estimating distances to prune distant candidates. During 
indexing, 
they construct a quantization codebook, assign each data vector to their nearest codebook vector, and store the codebook ids as compact quantization codes. 
At query time, they 
(1) pre-compute distances between the query and the codebook vectors, and (2) use these pre-computed distances to efficiently estimate query-object distances from the stored quantization codes, also known as quantized distances. In practice, these methods are often integrated with an inverted file (IVF) index~\cite{johnsonBillionscaleSimilaritySearch2019} to improve querying performance, with {IVF+RaBitQ}~\cite{gao2024rabitq} and {IVF+PQ}~\cite{jegouProductQuantizationNearest2011} being representative methods. The IVF index partitions the data vectors into $n_{cluster}$ clusters using the k-means algorithm during indexing and routes each query to the $n_{probe}$ nearest clusters at query time, thereby reducing the search space. Within these clusters, quantization methods execute their querying algorithm to obtain the final results.


\begin{figure*}[!t]
\begin{center}
\includegraphics[width=0.8\textwidth]{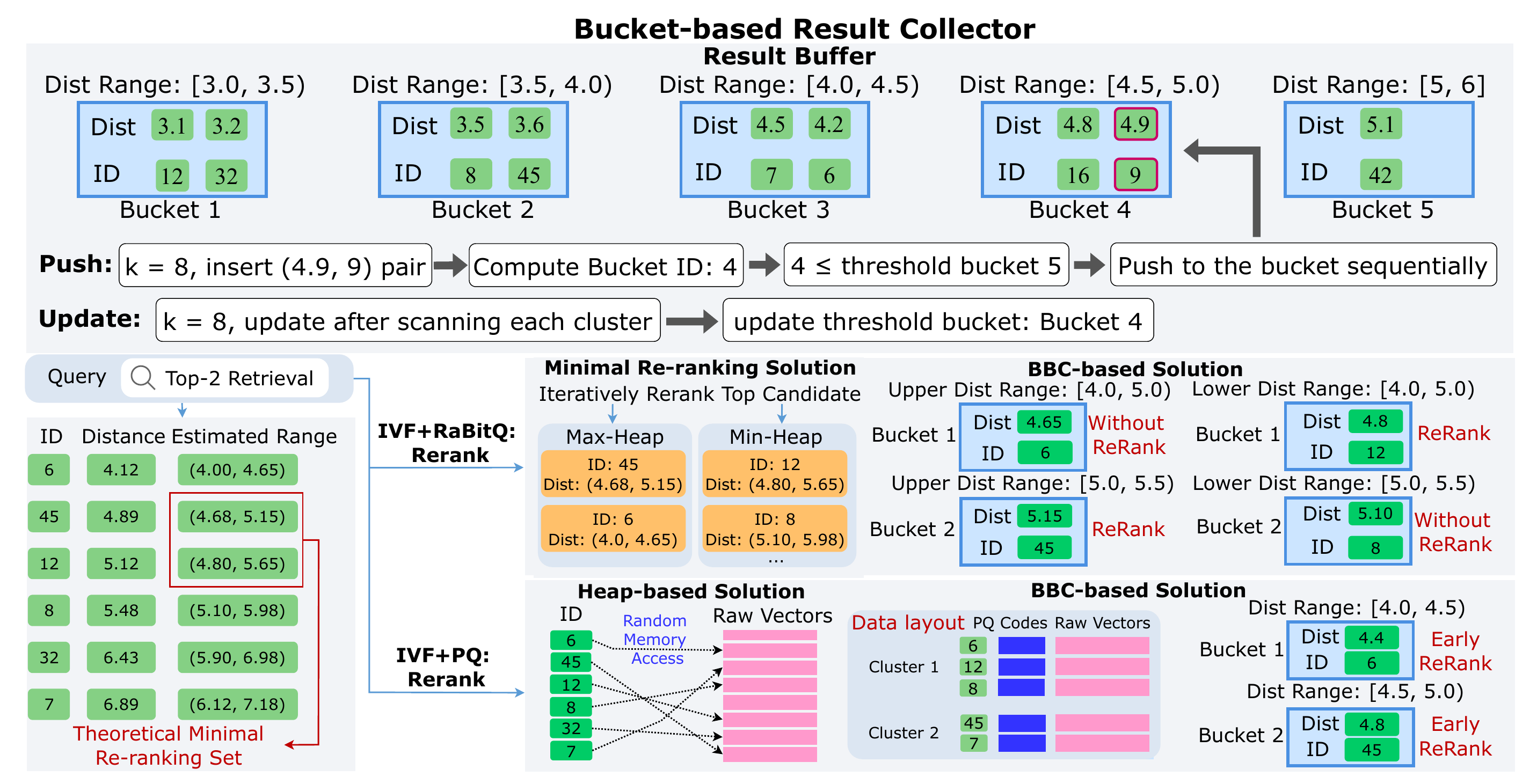}
\vspace*{-0.5em}
\caption{\revision{Illustration of the Proposed Bucket-based Result Collector.}}
\vspace*{-1.5em}
\label{fig:qs-framework}
\end{center}
\end{figure*}

Based on their pruning mechanisms and query processing strategies, quantization methods can be categorized into two types. 
In particular, {\textit{bounded}} methods such as RaBitQ~\cite{gao2024rabitq,rabitq2} {provide an estimated distance range with a high probabilistic guarantee (e.g., 99\%) and leverage these bounds for pruning. Specifically,} it employs a top-$k$ collector to maintain the currently found $k$ nearest neighbors and re-ranks any objects whose lower bounds fall below the collector's current threshold, as they may potentially enter the top-$k$ results. Because the objects stored in the collector often change rapidly during the early stage, the number of re-ranked objects is typically several times larger than $k$, as shown in Section~\ref{sec:qs-exp-results}. When $k$ increases, this results in a significantly higher re-ranking cost, as shown in Figure~\ref{fig:qs-exp-time-portion}. \textit{Unbounded} methods such as Product Quantization (PQ)~\cite{jegouProductQuantizationNearest2011} generally preset a hyperparameter $n_{cand} \gg k$ (e.g., $n_{cand}$=3,000 and $k$=100) to determine how many candidates are retrieved based on their estimated distances. When a query arrives, these methods employ a top-$k$ collector to gather $n_{cand}$ candidates according to their estimated distances, which are then re-ranked to produce the final results. Here, $n_{cand}$ is typically several times larger than $k$ to achieve high recall.
As $k$ increases, the number of re-ranked objects increases linearly, leading to a significant increase in the re-ranking cost, as shown in Figure~\ref{fig:qs-exp-time-portion}.

\noindent\textbf{Motivations.} While many ANN methods have been developed, no prior work has specifically investigated large-$k$ ANN queries. Therefore, we evaluate several representative ANN methods on large-$k$ ANN queries, 
where the top-$k$ collector is implemented using a binary-heap priority queue.
We 
present the evaluation results on the {C4} dataset \revision{in Figure~\ref{fig:qs-exp-motivation}}, which contains over 14 million passages. 
%
Figure~\ref{fig:qs-exp-time-portion} further presents a 
breakdown of time overheads under different $k$ based on VTune profiling. Details of the experimental setup are given in Section~\ref{sec:qs-exp-setup}. 
The experimental results reveal the performance degradation of these methods on large-$k$ ANN queries, and we have highlighted two major limitations of existing quantization-based methods in the Introduction. This motivate us to design a new algorithm for large-$k$ ANN queries.

%% file: framework.tex
\section{The BBC Method}
\label{sec:qs-method}
\subsection{Overview}
\label{sec:qs-overview}
In this section, we propose a novel bucket-based result collector ({\bbc}), composed of two new components: a bucket-based 
result buffer serving as the top-$k$ collector and two re-ranking algorithms.  
We proceed to give an overview of the two components.


First, the result buffer partitions the estimated distance range between the query and data objects into non-overlapping sub-ranges using one-dimensional quantization, as shown in Figure~\ref{fig:qs-framework}. Each sub-range corresponds to a bucket that contains two linear buffers, which sequentially store the IDs and distances of candidates. 
As presented in the Introduction, this design (1) provides bucket-level ranking of candidates based on their distances to the query, where candidates are ordered across buckets but remain unordered within each bucket (e.g., objects in Bucket 1 always have smaller distances to the query than those in Bucket 2). This design differs from binary heaps, which maintain a strict bucket-level order throughout, thereby reducing the ranking cost; 
and (2) exploits the hardware prefetching mechanism to reduce L1 cache misses. {This is because only the tail of the recently accessed linear buffer typically resides in the L1 cache, while the preceding elements are evicted, substantially alleviating L1 cache pressure.
}



%

Leveraging both bucket-level ranking and the number of objects stored in each bucket, we can efficiently identify the buckets that contain the exact top-$k$ results and the threshold bucket that holds boundary candidates near the threshold distance. 
For example, in Figure~\ref{fig:qs-framework}, when $k$ = 8, bucket 5 serves as the threshold bucket because the cumulative number of objects in the first five buckets reaches 8 (before inserting the object 9 with distance 4.9). 
The candidates within these buckets together form a \textit{candidate superset}. The upper bound distance of the threshold bucket serves as the \textit{relaxed threshold}, which is 6 in this case. 
%
%
At query phase, both can be efficiently updated. For example, 
when inserting a new object (e.g., object 9 with distance 4.9 in Fig.~\ref{fig:qs-framework}), 
we compute its bucket ID based on its distance to the query, compares it with the threshold bucket ID (e.g., Bucket 5), and appends the distance and ID to its corresponding bucket (e.g., pushing object 9 to Bucket 4).
%
%
After insertion, the threshold bucket can be updated by traversing buckets in order until the accumulated number of candidates reaches $k$, 
as the buckets are organized by distance range. {In Figure~\ref{fig:qs-framework}, when $k$ = 8, the threshold bucket shifts from bucket 5 to bucket 4 after inserting object 9.} 
The more distant buckets (e.g., Bucket 5) are no longer visited and are implicitly dropped without incurring additional cost. 
%
Finally, we only need to select a subset of objects from the threshold bucket and combine them with the objects in the preceding buckets to obtain the exact top-$k$ results.

Second, two re-ranking algorithms are designed for different types of quantization approaches to accelerate their re-ranking process, as illustrated in Figure~\ref{fig:qs-framework}. For \textit{bounded} methods, we establish 
a \revision{theoretical minimal re-ranking set} 
that minimizes the number of re-ranked objects without accuracy loss and propose a solution to achieve it. However, this solution incurs considerable heap overhead. To mitigate this, we propose a re-ranking algorithm built on our result buffer. 
It skips re-ranking objects that are definitely either inside 
or outside the top-$k$ 
and re-ranks only the uncertain ones near the threshold bucket. 
For \textit{unbounded} methods, we propose a novel early re-ranking algorithm that re-ranks many objects predicted to enter the re-ranking pool when accessing their data, 
reducing L1 cache misses from random memory access. 


In the rest of this section, we present the two components of {\bbc}: (1) the 
result buffer (Section~\ref{sec:result-buffer}) and (2) the two re-ranking algorithms
, along with the improved search algorithm (Section~\ref{sec:rerank}).

\subsection{Result Buffer}
\label{sec:result-buffer}
{We now describe how the result buffer partitions the estimated distance range between the query and objects into $m$ non-overlapping sub-ranges, as defined by the codebook $C$:}
\begin{equation}
C = \{ c_1, c_2, \ldots, c_{m+1} \}, \quad c_i < c_{i+1}.
\end{equation}
Accordingly, the sub-ranges are formally defined, {each corresponding to a bucket $B[i]$:}
\begin{equation}
{B}[i] = [c_i, c_{i+1}), \quad i = 1,2,\ldots,m.
\end{equation}
For each input distance–ID pair, the result buffer {determines the corresponding bucket by locating the interval in which the distance falls. Specifically, an object $o_i$ is assigned to bucket ${B}[j] $ if its distance from the query, $\mathrm{Dist}(q, o_i)$, satisfies $c_j \leq \mathrm{Dist}(q, o_i) < c_{j+1}$. The distance $\mathrm{Dist}(q, o_i)$ can refer to either the exact distance $\mathrm{Dist}_{\mathrm{exact}}$ or the estimated distance $\mathrm{Dist}_\mathrm{quant}$. 
Next we turn to the three core operations of the result buffer: \texttt{Push},  \texttt{Update}, and the \texttt{Collect} function, {which are used to collect the top-$k$ results.}
}


{The \texttt{Push} function in Algorithm~\ref{alg:result-buffer-workflow} (lines 1-4) details the procedure for inserting an object into the result buffer. Specifically, when a new object is inserted, the result buffer first computes its distance to the query and determines the corresponding bucket ID (line 2). It then compares this bucket index with the threshold bucket ID (line 3) and appends the distance–ID pair to 
the assigned bucket if the index does not exceed the threshold bucket ID (line 4). Here, using the threshold bucket ID for comparison essentially treats the upper bound of the threshold bucket's distance range as the relaxed threshold. 
In addition, bucket ID comparisons can benefit from fast SIMD-based integer comparison instructions, 
enabling simultaneous comparison of batches of objects. 
The quantization code computation 
can also be accelerated using SIMD instructions, as discussed later.}

{The \texttt{update} function in Algorithm~\ref{alg:result-buffer-workflow} (lines 5-11) describes the process of updating the threshold bucket. In particular, the buckets in the result buffer are arranged in ascending order of distance range, 
as shown in Figure~\ref{fig:qs-framework}. We accumulate the number of candidates in buckets in order until the total number reaches or exceeds $k$ (lines 6-9). Once this condition is met, the visited bucket is identified as the threshold bucket and returned (line 10). If the total number of objects stored in the result buffer is less than $k$, the threshold bucket is set to $\infty$, allowing all objects to be accepted (line 11). Since only dozens of buckets are involved, the update cost is negligible. 
Once the threshold bucket is updated, the more distant buckets are no longer accessed and are implicitly dropped, thus incurring no additional time cost for candidate maintenance. 
}

\begin{algorithm}[!t]
\small
\setcounter{AlgoLine}{0}
\LinesNumbered
\caption{The Workflow of Result Buffer}
\label{alg:result-buffer-workflow}
\SetKwFunction{Push}{\texttt{Push}}
\SetKwFunction{Update}{\texttt{Update}}
\SetKwFunction{Collect}{\texttt{Collect}}

\SetKwProg{Fn}{Function}{:}{}
\KwIn{Result buffer $B$; a set of clusters $cl$ to be scanned; retrieval size \(k\);}
\KwOut{The result buffer \(B\)}

\Fn{\Push{$q$, $o_i$, $\tau$, $B$}}{
    Compute \(\mathrm{Dist}(q, o_i)\) and bucket ID $j$\;
    \If{$j \leq \tau$}{
        Append (dist, $o_i$) to the tail of ${B}[j]$\;
    }
}

\Fn{\Update{$B$, $k$}}{
    $s \gets 0$\;
    \ForEach{$B[i] \in B$}{
        $s \gets s + |B[i]|$\;        
        \If{$s > k$}{
            \Return{$i$}\;
        }
    }
    \Return{$\infty$}\;
}

\Fn{\Collect{$q$, $cl$, $B$, $k$}}{
    Initialize threshold bucket $\tau \gets \infty$\;
    Construct codebook $C$ for $B$\;
    \ForEach{$cr \in cl$}{
        \ForEach{$o_i \in cr$}{
            \Push{$q$, $o_i$, $\tau$, $B$}\;
        }
        $\tau \gets \Update(B, k)$\;
    }
    $Res \gets \varnothing$\;
    \For{$i \gets 0$ \KwTo $\tau-1$}{
        $Res \gets Res \cup B[i]$\;       
    }
    $s \gets k - |Res|$\;
     Select the top-$s$ candidates from $B_{\tau}$ and append to $Res$ \;
    \Return{$Res$}\;
}
\end{algorithm}

{The \texttt{collect} function in Algorithm~\ref{alg:result-buffer-workflow} describes the workflow of 
collecting top-$k$ results based on (estimated) distance in IVF-based ANN methods. Specifically, we first initialize the threshold bucket to $\infty$ and construct the codebook $C$ for the result buffer, whose generation will be discussed later (lines 13–14). Objects are then inserted into the result buffer $B$ within each cluster using the \texttt{Push} function (lines 15–17). The threshold bucket is updated after processing each cluster using the \texttt{Update} function (line 18). This update is performed once per cluster because the relaxed threshold is very close to the exact threshold (as will be discussed later), and updating once per cluster helps to reduce the update cost. 
Once all clusters have been processed, the objects in the buckets preceding the threshold bucket are inserted into the result set $Res$ (lines 19-21). This is because the 
bucket-level ranking of candidates ensures that these objects are closer to the query than those in the subsequent buckets, thus falling within the top-$k$ candidates. Finally, we compute the number of objects $s$ that need to be selected from the threshold bucket (line 22), choose the top-$s$ objects from it, and add them to $Res$ (line 23). $Res$ now contains the exact top-$k$ results and is returned (line 24). This design substantially reduces the cost of maintaining exact top-$k$ results, as the selection operation is performed only once in the final stage within a single bucket.}


\noindent\textbf{{Deciding the Number of Buckets $m$.}}
A key consideration is how to determine the number of buckets $m$. If $m$ is too large, the increased number of vectors to be written raises the L1 cache miss rate. 
If $m$ is too small, objects are concentrated in just a few vectors, 
resulting in a costly final selection process (as shown in Section~\ref{sec:qs-exp-results}). 
\revision{To balance these factors, we aim to maximize $m$ based on the 
L1 cache capacity $C_{L1}$, while accounting for the space required by quantization codes $C_{quant}$ and lookup tables $C_{lut}$. Accordingly, the number of buckets $m$ is given by the following equation:}
\begin{equation}
m = \frac{C_{L1} - C_{quant} - C_{lut} }{256},
\label{eq:select-m}
\end{equation}
\revision{where $C_{quant} = 2 \times 32 \times \frac{d}{M} \times \frac{B}{8}$ is the space for quantization codes of the current and subsequent processed batches. Here, $d$ is the dimensionality, $M$ is the number of sub-vectors, and $B$ denotes the bits per sub-vector. $C_{lut} = \frac{d}{M} \times 2^B$ is the size of the lookup table. The denominator 256 reflects hardware prefetching considerations: 
the result buffer maintains two separate linear buffers for IDs and distances, and modern hardware typically prefetches two 64-byte cache lines for sequential access, we reserve $m \times 2 \times 2 \times 64 = 256m$ bytes to ensure the active tails of all buckets reside in the L1 cache. Notably, since we incorporate hardware parameters into Equation~\ref{eq:select-m}, the parameter $m$ will adjust accordingly when hardware settings change. Moreover, since Equation~\ref{eq:select-m} accounts for the memory overhead of the ANN algorithm and not all buckets are frequently accessed, small variations in $m$ do not lead to significant increases in L1 cache misses or latency (see Exp-6).} 

\noindent\textbf{{Deciding the Codebook $C$.}} 
\revision{We now describe} how to \revision{construct} the codebook \revision{$C = \{ c_1, c_2, \ldots, c_{m+1} \}$. 
The codebook is designed to satisfy two critical properties:} (1) Precision, ensuring that the relaxed threshold \revision{(e.g.,  $c_i$)} remains close to the precise value. This is crucial because it affects the efficiency of the final selection, and significant deviations could make the process time-consuming; (2) Efficiency, ensuring low latency in both codebook generation and quantization code computation. Since the result buffer serves as the top-$k$ collector in the ANN search, slow generation and computation would offset its benefits. Note that the codebook requires dynamic generation for each query, rather than pre-computation, as query-object distance distributions vary across queries.

\revision{To quantify precision, we formalize it as the total quantization error. Given a query $q$, for each object $o_i \in D$, let $\mathrm{Dist}(q, o_i)$ denote its distance to the query and $a_i$ denote its assigned quantization code. Each code corresponds to a bucket whose upper boundary is $c_{a_{i+1}}$. The precision objective $Cost(C, D)$ is then defined as:}
\vspace*{-0.5em}
\begin{equation}
{Cost(C, D) = \sum_{i=1}^N \left|c_{a_{i+1}} - \mathrm{Dist}(q, o_i)\right|},
\label{eq:cost}
\end{equation}
The \revision{corresponding} optimal solution \revision{$\{C, D\}$} is given by:
\begin{equation}
\{C, D\} = \arg\min_{\{C, D\}} Cost(C, D).
\label{eq:objective-1}
\end{equation}
The problem is then to find a centroid codebook $C$ that minimizes Equation~\ref{eq:cost}. Because computing the distance for all objects is impractical, an exact solution to this problem becomes infeasible. The two common approximate methods for one-dimensional quantization are equal-depth partition and equal-width partition, which are described as follows:

\begin{itemize}[leftmargin=*, topsep=0pt]
  \item \textbf{Equal-depth partition}~\cite{muralikrishna1988equi}: This method divides the data range into intervals with equal data points, resulting in non-uniform intervals based on the data distribution. While it maximizes bucket utilization and has higher precision, it involves a slower generation and computation process.
  
  \item \textbf{Equal-width partition}~\cite{piatetsky1984accurate}: This method divides the data range into intervals of equal width, regardless of the data distribution. 
  While it enables faster codebook generation and code computation, it suffers from lower precision 
  because it does not adapt to the data distribution.
\end{itemize}
First, we demonstrate that the precision level of equal-{depth} quantization meets our requirements. In particular, in high-dimensional space, the distance concentration phenomenon~\cite{curse1, curse2} causes distances between vectors to concentrate around a value \revision{with only a small deviation. Based on this, 
we provide (1) a quantitative analysis 
with a proof sketch, with the full proof deferred to the technical report~\cite{fullversion} and (2) experimental validation on real-world high-dimensional datasets (see Exp-4). }

\begin{figure}[!t]
\centering
\includegraphics[width=0.75\columnwidth]{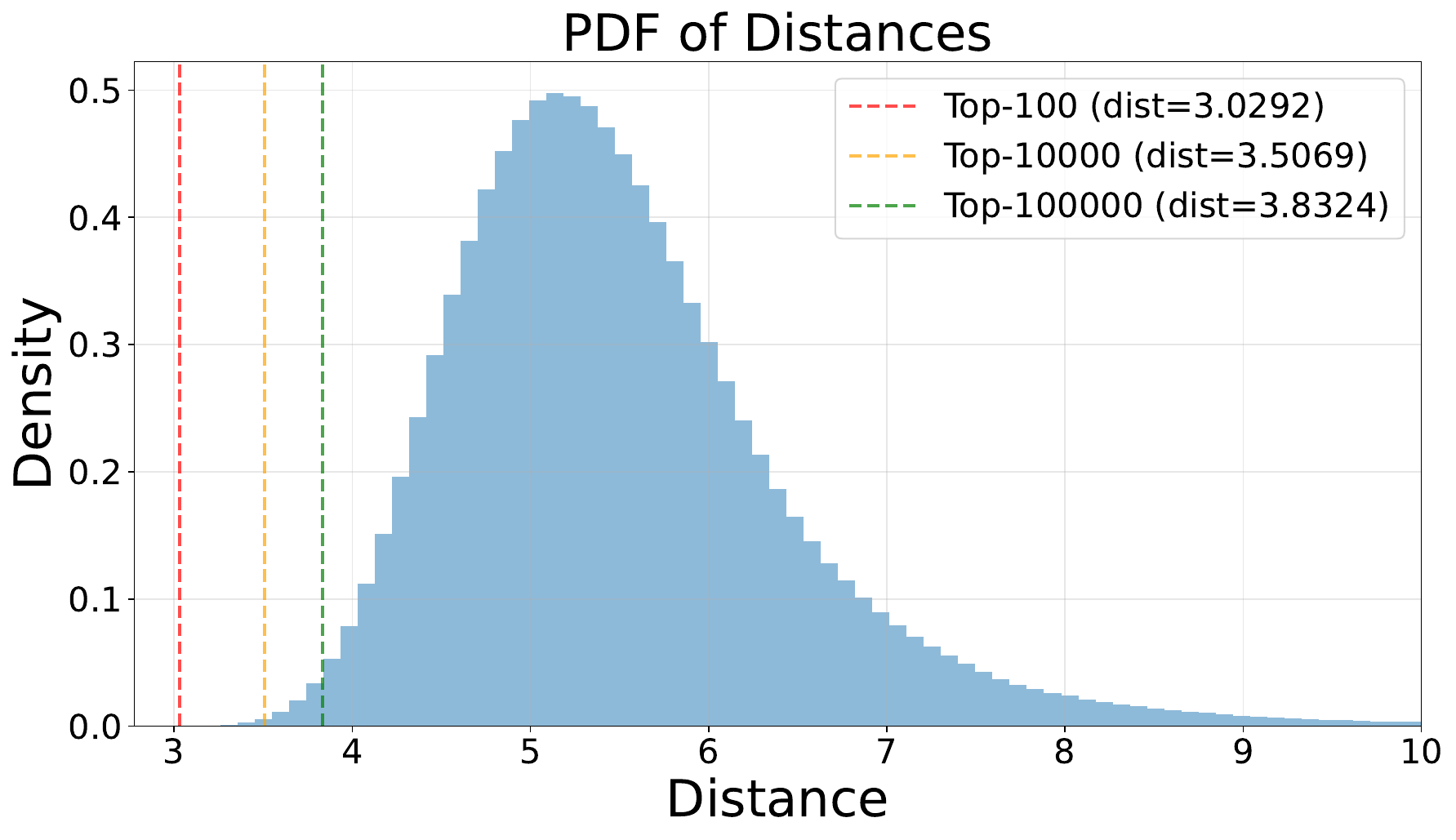}
\vspace*{-1em}
\caption{\revision{The probability density function (PDF) of distances between the query and data vectors on the C4 dataset.}}
\vspace*{-1.5em}
\label{fig:qs-pdf}
\end{figure}

\vspace*{-0.5em}
\begin{theorem}[\revision{Expected Mean Absolute Error}]
\label{th:error}
\revision{Let $q,o \in \mathbb{R}^d$ be random vectors with norms satisfying $\alpha < \|q\|_2, \|o\|_2 < \beta$, and let $R=\|q-o\|_2 \in [0,2\beta]$ denote their Euclidean distance. Assume that $R$ has a sub-Gaussian left tail around a concentration value $\mu \in [\alpha, 2\beta]$, i.e., there exist absolute 
constants $c_1,c_0>0$ 
such that for any $x \le \mu$,}
\[
\mathbb{P}(R \le x) \le 2c_1 \exp(-c_0 d (\mu-x)^2).
\]
\revision{For an integer $m \ge 2$, let $b_1,\dots,b_{m+1}$ be an equal-depth partition of $[0,\mu]$ with $b_1=0$ and $b_{m+1}=\mu$, satisfying}
\[
\mathbb{P}(R \le b_i) = \frac{i-1}{m}\,\mathbb{P}(R \le \mu), \quad i=1,\dots,m+1.
\]
\revision{Define the quantized distance $\widehat{R}$ by mapping $R \in [b_i,b_{i+1}]$ to $b_{i+1}$. Then the expected absolute quantization error over $[0,\mu]$ satisfies}
\[
\mathbb E|R-\widehat R|
\;\le\;
c_1 \sqrt{\frac{\pi}{c_0 d}}
\;+\;
\sqrt{\frac{\log(2 c_1 m)}{c_0 d}}.
\]
\end{theorem}
\begin{proof}[\revision{Proof Sketch}]
{The key idea is to view the quantization error as the area between two cumulative distribution functions, i.e., the integral of the difference between the true CDF and the quantized CDF, which can be represented as $\int_{0}^{\mu} \Delta(x)\,dx$, where $\Delta(x)=\mathbb{P}(R \le x)-\mathbb{P}(\widehat{R} \le x)$.}
{We partition the integration regions $[0, \mu]$ into two regions: $[0,z]$ and $[z,\mu]$, where the transition point is chosen as $z = \mu - \sqrt{\frac{\log(2 c_1 m)}{c_0 d}}$ according to sub-Gaussian concentration. }
{For $[0,z]$, the integral is bounded by the Gaussian integral}
\[
\begin{aligned}
\int_{0}^{z} \Delta(x)\,dx
&\le \int_{0}^{z} \mathbb{P}(R \le x)\,dx
\le 2c_1 \int_{-\infty}^{\mu} e^{-c_0 d(\mu-x)^2}\,dx \\
&= c_1 \sqrt{\frac{\pi}{c_0 d}} .
\end{aligned}
\]
{For $[z,\sqrt{2}]$, we use a worst-case bound of $1$ on the CDF difference, so the contribution is bounded by the interval width $\sqrt{\frac{\log(2c_{1}m)}{c_0d}}$.}

{Summing the two contributions yields the stated bound.}
\end{proof}

\noindent\revision{\textbf{Extension of Theorem~\ref{th:error} to real-world datasets.} First, when $d$ ranges from hundreds to thousands and $m$, computed by Equation~\ref{eq:select-m}, is typically on the order of dozens, the bound is on the order of $10^{-2}$, which is sufficiently tight. Second, the assumption that $R$ has a sub-Gaussian left tail around a concentration value $\mu \in [\alpha, 2\beta]$ follows from the distance concentration phenomenon in high-dimensional spaces~\cite{gao2024rabitq, curse1, curse2,aggarwal2001surprising}, and the phenomenon is widely observed in real-world datasets. We further illustrate this phenomenon on the C4 dataset. As shown in Figure~\ref{fig:qs-pdf}, the distances between the query and data vectors concentrate around 5.2, with most values lying in $[4,8]$. Third, we restrict quantization to the interval $[0,\mu]$ because the result buffer is designed to collect the top-$k$ results. We therefore target on estimating the distance range between the query and the top-$k$ data vectors, and apply equal-depth quantization only to this subset rather than the entire dataset. Figure~\ref{fig:qs-pdf} illustrates this: distances between the query and the top-100,000 candidates (e.g., 3.8324) are significantly smaller than the concentration distance (e.g., 5.2), implying that the derived theoretical bound serves as a loose upper bound. This also indicates that, even under highly skewed data distributions, including extreme outliers far from the query, the tight bound remains unaffected, since such points are excluded from the equal-depth quantization process. We empirically validate the above bound on real-world datasets, showing that the gap between the quantized values and the true threshold values is on the order of $10^{-2}$ or smaller (see Exp-4).} 

\noindent\textbf{Codebook Generation Based on Estimated Distance.} We now detail how to construct the codebook $C$ for the result buffer. 
{
To \revision{efficiently estimate the distance range of the result buffer}, we sample a subset of the dataset, denoted as $D_{\text{sample}}$, which typically consists of tens of thousands of objects, and quickly compute their estimated distances. In practice, $D_{\text{sample}}$ is formed using objects from the 5–10 nearest clusters.} 
We then perform a partial sort on the estimated distances to obtain the top-$k$ results. 
Since the partial sort is performed only once, its computational cost is negligible. Afterwards, we derive the minimum distance $d_{\min}$ and maximum distance $d_{\max}$ from the \revision{local top-$k$ results. 
It is evident that the $d_{\max}$ are necessarily farther than the true top-$k$ estimated distance. Therefore, the estimated distance range will not overflow at $d_{max}$, any potential overflow can only occur at $d_{min}$, which rather rare and can be safely handled by the boundary control, as described below.}  To maintain computational efficiency, we follow the approach of~\cite{mousavi2011fast}, which first computes $n_{ew}$ equal-width buckets and then reassigns these buckets into $m$ equal-depth buckets, where a lookup table $\text{map}$ is used to preserve the correspondence between them.

\noindent\textbf{Quantization Code Computation.} For newly inserted objects, we first compute their equal-width codes and subsequently clamp those outside the range $[0,m]$ to prevent boundary overflow, and then obtain the corresponding mapping ID via the lookup table, as formulated below:
\begin{equation}
\label{eq:remap}
a_i = 
  \text{map}\left[\mathrm{clamp}\!\left(
    \left\lfloor \frac{\mathrm{Dist(q,o_i)} - d_{\min}}{\delta} \right\rfloor,
    0, m\right]
  \right)
\end{equation}
\revision{where $\mathrm{clamp}(x, y, z) = \max(y, \min(x, z))$.} SIMD instructions can be employed to accelerate the computation of $\left\lfloor \tfrac{\mathrm{Dist(q,o_i)}-d_{\min}}{\delta} \right\rfloor$ \revision{as well as the $\mathrm{clamp}$ operation}, enabling the batch processing of dozens of objects. 
We fix $n_{\text{ew}}$ to 256, so that the mapping can be stored using {uint8}. {This has two advantages: (1) The mapping requires only 256 bytes in total, allowing it to reside in L1 cache with minimal memory overhead; (2) AVX instructions can process 4$\times$ as many {uint8} values per instruction as {float32} values, which speeds up batch comparison.} 


\noindent\revision{\textbf{Complexity Analysis.} The time complexity of the \texttt{Push} operation is $O(1)$, dominated by the bucket ID computation, which is significantly more efficient than the $O(\log (k))$ complexity of binary heap. The \texttt{update} operation has a time complexity of $O(m)$. Given that  $m$ computed by Equation~\ref{eq:select-m} is typically dozens and the operation is performed periodically, its amortized time complexity can be considered O(1), which is more efficient than the $O(\log(k))$ cost of binary heap. The \texttt{collect} operation has a time complexity of $O(N)$. Since it is invoked only once at the end of the query process, its overall cost is negligible.}



\subsection{Re-rank Algorithms}
\label{sec:rerank}
We now discuss how to integrate the result buffer with existing quantization-based methods:  {IVF+RaBitQ} and {IVF+PQ}. 
We aim to enhance their efficiency for handling large-$k$ ANN queries.
To achieve this, we 
introduce several new techniques. 

\begin{algorithm}[!t]
\small
\caption{\revision{Minimal Re-ranking solution of IVF+RaBitQ}}
\label{alg:RabitQ-optimal}
\KwIn{Query $q$, retrieval size $k$, and the object set $O$ to be scanned.}
\KwOut{Top-$k$ results}
Initialize max-heap $H_u$ and min-heap $H_l$\;

\ForEach{$o \in O$}{
    Compute the lower and upper bounds $o_{lb}, o_{ub}$ from $o$ to $q$\;
    \If{$o_{ub} < H_u.\mathrm{top}_{ub}$}{
        Insert $(o, o_{lb}, o_{ub})$ into $H_u$\;
        \If{$|H_u| > k$}{
            Move top item from $H_u$ to $H_l$\;
        }
    }
    \ElseIf{$o_{lb} < H_u.\mathrm{top}_{ub}$}{
        Insert $(o, o_{lb})$ into $H_l$\;
    }
}

$Vis \leftarrow \varnothing$\;
\While{$H_u.\mathrm{top}_{ub} > H_l.\mathrm{top}_{lb}$}{
    Select $o$ from $H_u$ or $H_l$ with smaller lower bound\;
    Compute exact distance of $o$ and insert into $H_u$\;
    $Vis \leftarrow Vis \cup \{o\}$\;
    \If{$|H_u| > k$}{
        Move top item from $H_u$ to $H_l$ if not visited\;
    }
}

\Return{$H_u$}\;
\end{algorithm}

\noindent\textbf{Integrating with {IVF+RabitQ}}. 
We introduce a novel re-ranking algorithm based on our result buffer by exploiting the bound property~\cite{gao2024rabitq}, namely, the true distance fall within the estimated bound with high probability (e.g., 99\%). 
This property enables us to efficiently estimate the distance range between each visited data object and the query. 
As a result, it is unnecessary to re-rank data objects that are guaranteed to be within the top-$k$ 
or definitively outside the top-$k$. Re-ranking is required only for ambiguous candidates whose inclusion in the top-$k$ remains uncertain. 
Based on this insight, we formally define the 
minimal re-ranking scenario. 


\begin{observation}[Minimal Re-ranking Scenario]
\label{def:optimal-rabitq}
$\mathrm{Dist}_k$ denotes the distance \revision{between the query and the boundary top-$k$ object}. When the objective is to minimize the number of objects to be re-ranked, the {minimal re-ranking scenario} for 
bounded quantization method is to re-rank only those objects $o \in D$ whose lower and upper bounds satisfy $[o_{lb}, o_{ub}] \cap \{\mathrm{Dist}_k\} \neq \varnothing$. This set represents the theoretical minimal set of objects that must be re-ranked without accuracy loss. 
\end{observation}
\revision{Figure~\ref{fig:qs-framework} illustrates this scenario, where we aim to retrieve the top-2 results. Objects 6 and 45 are the true top-2 results, where object 45 serves as the boundary top-$k$ item determining inclusion. The estimated range presents the estimated distance ranges based on RaBitQ. When these ranges overlap with the boundary distance (e.g., objects 45 and 12), the corresponding candidates must be re-ranked, forming the theoretical minimal re-ranking set.}



\begin{algorithm}[!t]
\small
\setcounter{AlgoLine}{0}
\LinesNumbered
\caption{Improved Search Algorithm of {IVF+RaBitQ}}
\label{alg:RabitQ-greedy}
\KwIn{Query $q$, retrieved size $k$, and the clusters to be scanned $cl$.}
\KwOut{Top-$k$ results.}

Initialize two result buffers ${B}_u$ and ${B}_l$\;
Generate codebook $C$ for ${B}_u$ and ${B}_l$\;

\ForEach{$cr \in cl$}{
    \ForEach{$o_i \in cr$}{
        Compute the lower/upper bounds $o_{lb}$/$o_{ub}$ for $o_i$ and their respective quantization codes $a_{lb}$, $a_{ub}$\;

        \If{$a_{ub} < \tau$}{
            Insert $o_i$ into ${B}_u$\;
        }
        \ElseIf{$a_{lb} < \tau$}{
            Insert $o_i$ into ${B}_l$\;
        }
    }
    $\tau \leftarrow \texttt{Update}({B}_u, k)$\;
}

Insert objects from ${B}_u$ into ${B}_l$\;

$i \leftarrow 0$, \quad $j \leftarrow \tau$\;
Initialize result buffer ${B}_\mathrm{exact}$ with codebook $C$\;

\While{$i < j$}{
    \ForEach{$o \in {B}_l[i] \cup {B}_u[j]$}{
        Compute exact distance ${Dist}_{\mathrm{exact}}(q, o)$\;
        Insert $o$ into ${B}_{\mathrm{exact}}$ based on ${Dist}_{\mathrm{exact}}(q, o)$\;
    }
    Clear ${B}_l[i]$ and ${B}_u[j]$\;
    $i \leftarrow 0$, $j \leftarrow$ index of the threshold bucket in ${B}_u \cup {B}_\mathrm{exact}$\;
    
    \While{${B}_l[i]$.empty()}{
        $i \leftarrow i + 1$\;
    }
}

\Return the objects in $B_u \cup B_{\mathrm{exact}}$ as the final top-$k$ results\;
\end{algorithm}

\noindent\textbf{Solution to the Minimal Re-ranking Scenario.} 
We design Algorithm~\ref{alg:RabitQ-optimal} to achieve the minimal re-ranking scenario described in Observation~\ref{def:optimal-rabitq}. 
It consists of two phases: candidate collection phase and re-ranking phase. In the candidate collection phase, \revision{it maintains the top-$k$ candidates with the smallest upper bounds in a max-heap $H_u$ , while storing objects whose lower bounds fall below the $k$-th upper bound in a min-heap $H_l$ (lines 2-9). In the re-ranking phase, a set $Vis$ tracks objects whose exact distances have been computed (line 10). At each step, the algorithm selects the object with the smaller lower bound from the tops of the two heaps for exact distance evaluation, and updates the heaps accordingly (lines 11–16). The process terminates when the largest upper bound in $H_u$ is smaller than the smallest lower bound in $H_l$ (line 13). The correctness proof is provided in the technical report~\cite{fullversion} due to page limitations. Figure~\ref{fig:qs-framework} gives an example where objects 45 and 12 are iteratively re-ranked. The process terminates when objects 6 and 8 reach the top, as the lower bound of object 8 exceeds the upper bound of object 6. However, maintaining a max-heap of size $k$ and an unbounded min-heap incurs prohibitively high overhead when $k$ is large. As a result, this approach can even be slower than IVF+RaBitQ in practice (Section~\ref{sec:qs-exp-results}).}
\begin{algorithm}[!t]
\small
\setcounter{AlgoLine}{0}
\LinesNumbered
\caption{Improved Search Algorithm of {IVF+PQ}}
\label{alg:early-rerank-opq}
\KwIn{Query $q$, the number of objects to be retrieved $k$, the clusters to be scanned $cl$, the number of objects to be re-ranked $n_{cand}$.}
\KwOut{Top-$k$ results.}
Initialize result buffer $B$\;
Sample a subset $D_{\text{sample}} \subset D$\;


Produce codebook $C$ for $B$\;
Generate the predicted threshold bucket $\tau^{\mathrm{pred}}$\;

\ForEach{$cr \in cl$}{
    \ForEach{$o_i \in cr$}{
        Compute $\mathrm{Dist}_\mathrm{quant}(q, o_i)$ and bucket id $a_i$\;
        \If{$a_i < \tau$}{
            \If{$a_i < \tau^{\mathrm{pred}}$}{
                Compute exact distance $\mathrm{Dist}_\mathrm{exact}(q, o_i)$\;
                Insert $o$ and $\mathrm{Dist}_\mathrm{exact}(q, o_i)$ into $B$\;
            }
            \Else{
                Insert $o$ and $\mathrm{Dist}_\mathrm{quant}(q, o_i)$ into $B$\;
            }
        }
    }
    Update $\tau^{\mathrm{pred}}$ and $\tau$\; 
}
Re-rank remaining candidates in $B$\;
\Return{\textnormal{top-$k$ results}}\;
\end{algorithm}

\revision{To address this limitation}, we propose a greedy re-ranking algorithm based on our result buffer, which significantly reduces the number of items to be re-ranked.  \revision{Figure~\ref{fig:qs-framework} provides a example, where we replace the two heaps with two result buffers, and only re-rank items within some boundary buckets, thereby enabling early termination of the re-ranking process.} Algorithm~\ref{alg:RabitQ-greedy} details our proposed re-ranking approach and
summarizes the enhanced search algorithm. Specifically, instead of using two heaps, we replace the two heaps with two result buffers that share the same codebook (lines 1–2). Next, we collect candidates based on the lower and upper bounds of $o_i$. In particular, when the object's upper bound lies within the top-$k$ upper bounds, the object is inserted into $C_u$; otherwise, if its lower bound falls below the threshold, it is inserted into $C_l$, as it may still qualify for the final top-$k$ results (lines 3-9). For the collected candidates, we first re-collect the falsely dropped candidates into $C_l$ from $C_u$ (line 11). Then we greedily re-rank all items in the marginal buckets, that is, the top bucket of $\mathrm{C}_l$ and the threshold bucket of $\mathrm{C}_u$, and insert the results into a new result buffer $\mathrm{C}_\mathrm{exact}$ for storing items with exact distance (lines 15-17). After each computation iteration, we clear the candidate buckets and  update the marginal buckets until $i \geq j$ (lines 18-21). When $i \geq j$, the upper bound of $\mathrm{C}_u$ is smaller than the lower bound of $\mathrm{C}_l$, indicating that no further potential candidates exist. Finally, we return the results in $\mathcal{C}_u \cup \mathcal{C}_\mathrm{exact}$ as the final results. The experimental results show that this method achieves near-minimal re-ranking reduction, leading to a substantial decrease in re-ranking time (as shown in Section~\ref{sec:qs-exp-results}).



\noindent\textbf{Integrating with {IVF+PQ}}. Due to the unbounded nature of the {PQ} algorithm, we cannot reduce the number of objects to be re-ranked. Therefore, we propose an early re-ranking algorithm built upon our result buffer to reduce the cache misses caused by the random memory access patterns of {IVF+PQ}, \revision{as illustrated in Figure~\ref{fig:qs-framework}. In particular, we optimize the memory data layout to store each object's PQ code and embedding contiguously. When estimating the distance of an object from its PQ code, we predict whether it will enter the re-ranking pool based on the estimated distance. If so, we immediately compute its exact distance to reduce L1 cache misses.}

Algorithm~\ref{alg:early-rerank-opq} details our proposed re-ranking approach and summarizes the improved search algorithm. During the sampling stage of codebook generation (line 2), we use the bucket containing $\left(\frac{|O_\text{sample}|}{|O|} \times n_{cand}\right)$-th quantized distance as the predicted threshold bucket $\tau^{\mathrm{pred}}$ (line 3). Then, during the scanning phase, for each object in a bucket preceding the threshold bucket, we compute the exact distance if its bucket ID $a_i < \tau^{\mathrm{pred}}$ and insert the exact value; otherwise, we insert its quantized distance (lines 5-13). After scanning each cluster, we update the predicted threshold $\tau^{\mathrm{pred}}$ using the $\left(\frac{|O_\text{scanned}|}{|O|} \times n_{cand}\right)$-th quantized distance and update the threshold bucket $\tau$ as described in Algorithm~\ref{alg:result-buffer-workflow} (line 14). This approach leads to a substantial reduction in L1 cache misses and re-ranking time, as verified in Exp-5. A 
concern is that 
$\tau^{\mathrm{pred}}$ might be too large, causing unnecessary re-ranking. In practice, this does not occur because clusters are traversed from nearest to farthest based on query-centroid distance (line 5), which produces a distribution skewed toward smaller values, thereby keeping $\tau^{\mathrm{pred}}$ low.

\noindent\revision{\textbf{Extension to different Metrics.} 
Our proposed result buffer applies to various similarity measures, including metric distances such as Euclidean distance, as well as cosine similarity and inner product, as these measures exhibit the distance concentration phenomenon in high-dimensional spaces~\cite{aggarwal2001surprising} and our result buffer does not rely on specific properties of metric distances. Similarly, the two re-ranking algorithms are compatible with any distance measure supported by PQ, RaBitQ, or other ANN techniques, as they do not depend on specific properties of metric distances as well.}

%% file: exp.tex
\begin{table}[t]
    \small
    \centering
    \caption{Dataset Statistics}
    \label{tab:qs-dataset}
    \vspace*{-1em}
    \renewcommand{\arraystretch}{1.2}
    \setlength{\arrayrulewidth}{0.8pt}
    \begin{tabularx}{0.45\textwidth}{
        >{\centering\arraybackslash}p{1.4cm}
        >{\centering\arraybackslash}X
        >{\centering\arraybackslash}p{0.8cm}
        >{\centering\arraybackslash}p{0.8cm}
        >{\centering\arraybackslash}p{1.2cm}
        }
        \toprule
        \textbf{Dataset} & $|D|$ & $d$ & $|Q|$ & {Size (GB)} \\
        \midrule
        {WiKi} & 10,000,000 & 1,536 & 1,000 & 58 \\        
        {C4} & 14,252,691 & 1,024 & 1,000 & 54  \\
        {MSMARCO} & 18,000,000 & 768 & 1,000 & 51 \\
        {Deep100M} & 100,000,000 & 96 & 10,000 & 35 \\        
        \bottomrule
    \end{tabularx}
    \vspace*{-1.3em}
\end{table}

\section{Experiments}
\label{sec:qs-exp}

\subsection{Evaluation Setup}
\label{sec:qs-exp-setup}

\noindent\textbf{Datasets.} We conduct experiments on {four} real-world datasets\footnote{\yin{Note that our techniques introduce negligible additional space overhead. Therefore, we defer the memory cost analysis to the technical report~\cite{fullversion}.}}. {Specifically, we evaluate our method and baselines on the 
Wiki, C4, {MSMARCO}, and {Deep100M} datasets. 
} 
The statistics of these datasets are listed in Table~\ref{tab:qs-dataset}. Details of each dataset are stated in the technical report due to page limitations~\cite{fullversion}.

\begin{figure*}[thbp]
\centering
\subcaptionbox{{{{Wiki}}}\label{fig:exp-wiki-1536}}{
\includegraphics[width=0.9\textwidth]{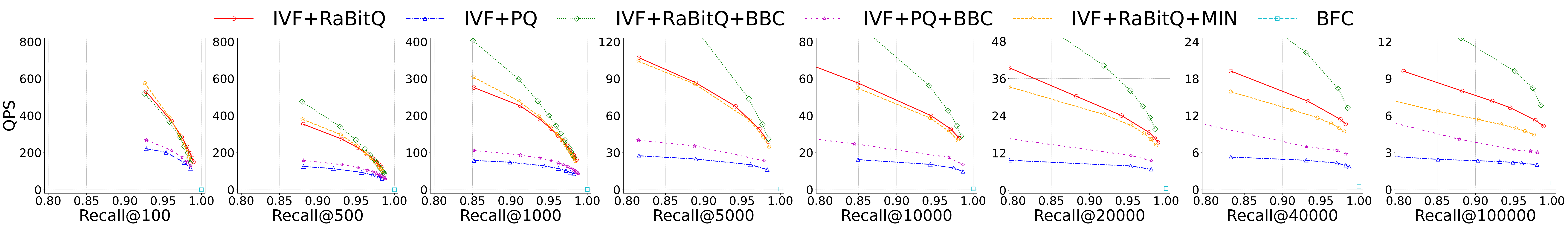}
}
\subcaptionbox{{{{C4}}}\label{fig:exp-c4-1024}}{
\includegraphics[width=0.9\textwidth]{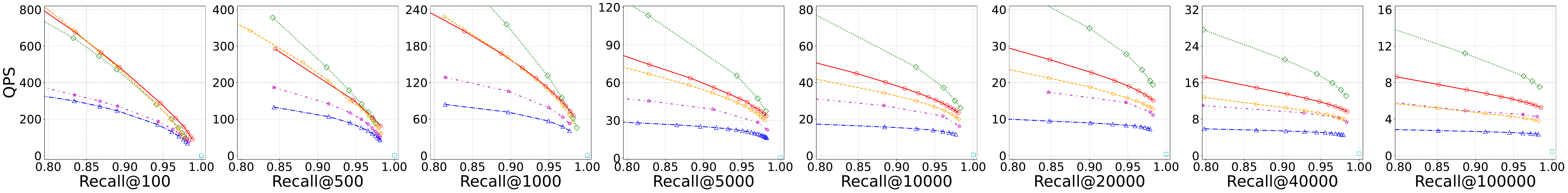}
}
\subcaptionbox{{{MSMARCO}}\label{fig:exp-marco}}{
\includegraphics[width=0.9\textwidth]{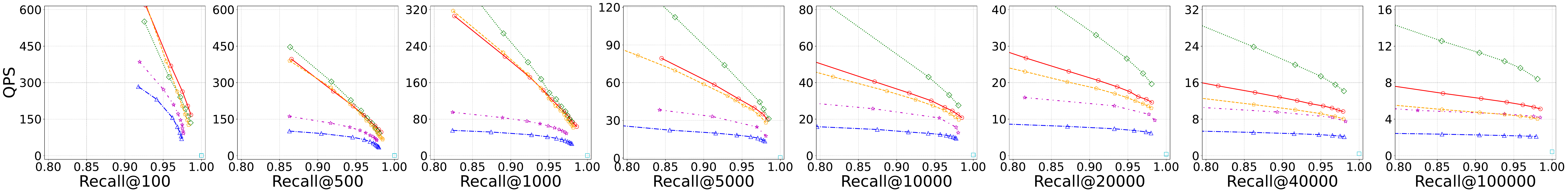}
}
\subcaptionbox{{{Deep100M}}\label{fig:exp-deep}}{
\includegraphics[width=0.9\textwidth]{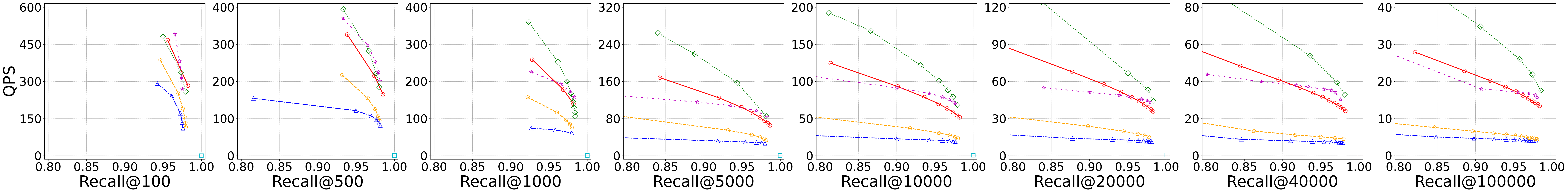}
}
\vspace*{-1em}
\caption{\revision{The accuracy-efficiency trade-off results under different $k$ (upper and right is better).}}
\vspace*{-1.5em}
\label{fig:exp-acc-eff}
\end{figure*}

\noindent\textbf{Baselines.} \underline{First}, 
we evaluate four representative ANN methods for large-$k$ ANN queries, \revision{as described in the Introduction and detailed in the technical report~\cite{fullversion}}. We integrate our proposed {\bbc} with existing quantization-based methods, {IVF+PQ} and {IVF+RabitQ}, yielding {IVF+PQ+\bbc} and {IVF+RaBitQ+\bbc}. We compare them with their original counterparts and also include the minimal re-ranking solution for {IVF+RaBitQ}, described in Section~\ref{sec:rerank} and denoted as {IVF+RaBitQ+MIN}, as a baseline. \revision{
We also include brute-force search, denoted as BFC, as a baseline for comparison.} {\underline{Second}}, to compare the efficiency of our proposed result buffer, denoted as {RB}, for collecting the top-$k$ results, we consider {four} baselines: the binary heap (denoted as {Heap}), {the cache-optimized d-ary Heap~\cite{johnson1975priority} (denoted as {d-Heap}), the sorted linear buffer used in~\cite{DBLP:journals/pvldb/AguerrebereBHTW23} (denoted as {Sorted}), }and the lazy update method (denoted as {Lazy}). In particular, Sorted maintains 
candidates in a sorted linear buffer, shifting the entire buffer upon each insertion. {Lazy} stores candidates whose distances to the query are below the current threshold in a linear buffer and updates both the buffer and the threshold using a 
{SIMD-optimized partial} sorting operation (e.g., 
{x86simdsort::qselect}), after each cluster is processed. 



\noindent\textbf{Evaluation metrics.} \underline{First}, for ANN query, 
we use recall rate recall$@k = \frac{\mathcal{R}\cap \tilde{\mathcal{R}}}{k}$~\cite{liApproximateNearestNeighbor2020,wang2021comprehensive} to evaluate the accuracy of search results and queries per second QPS $= \frac{|Q|}{t}$~\cite{fu2019fast} to evaluate the search's efficiency. Here, $\mathcal{R}$ represents the result retrieved by the \revision{ANN} index, $\tilde{\mathcal{R}}$ denotes the ground-truth result computed by the brute-force search. QPS$ = \frac{|Q|}{t}$~\cite{fu2019fast} is the ratio of the number of queries ($|Q|$) to the total search time ($t$), representing the number of queries processed per second. \revision{We also report the relative error between the distances of approximate nearest neighbors in $\mathcal{R}$ and true nearest neighbors in $\tilde{\mathcal{R}}$ to assess the quality of the retrieved results. 
The relative error~\cite{liApproximateNearestNeighbor2020,sun2014srs} is computed as $\frac{1}{k}\sum_{i=0}^{k}(\frac{\mathrm{Dist}(q, \mathcal{R}[i])}{\mathrm{Dist(q, \tilde{\mathcal{R}}[i])}}-1)$, where $q$ denotes the query, $\tilde{\mathcal{R}}[i]$ is the $i$-th true nearest neighbor of $q$ in $\tilde{\mathcal{R}}$, and $\mathcal{R}[i]$ is the $i$-th approximate nearest neighbor retrieved by the ANN algorithm.} 
\underline{Second}, to compare our proposed result buffer with {alternative approaches,} 
we evaluate the time overhead (milliseconds) of collecting top-$k$ results under varying dataset sizes and different values of $k$ using VTune profiling. 

\noindent\textbf{Parameter Settings.} Following the suggestion from Faiss~\cite{johnsonBillionscaleSimilaritySearch2019}, the number of clusters for {IVF}, {IVF+RaBitQ}, and {IVF+PQ} is set to approximately $\sqrt{|D|}$, which is 4,096 in our experiments. For {IVF+RaBitQ}, we use the default quantization parameters from the original paper~\cite{gao2024rabitq},  $\epsilon_0 = 1.9$ and $B_q = 4$. For {IVF+PQ}, the number of sub-vectors $M$ is set to  $\frac{d}{4}$, and the number of bits per sub-vector is set to 4, {resulting in $B = d$}, following the original settings~\cite{DBLP:conf/mir/AndreKS17,johnsonBillionscaleSimilaritySearch2019, DBLP:conf/icml/GuoSLGSCK20}. For {HNSW}, during indexing, we set the candidate list size $ef_{\text{construction}}$ and the maximum number of edges per node $M$ to 200 and 32, respectively, following previous studies~\cite{malkovEfficientRobustApproximate2020}. To evaluate the capability of methods in handling large-$k$ ANN queries under different $k$, we vary $k$ from 5,000 to 100,000, using five representative values (5,000, 10,000, 20,000, 40,000, and 100,000), which are widely used in practical scenarios~\cite{lee2023rethinking,gao2021learning}. Additionally, we also present experimental results for $k$ values ranging from \revision{100 to 1,000}, showing that {\bbc} does not slow down existing methods for small $k$. At query time, {IVF} and {IVF+RaBitQ} vary the number of clusters for routing $n_{probe}$ from 10 to 1,200 and {IVF+PQ} use the same $n_{probe}$ as {IVF+RaBitQ}. For {IVF+PQ}, for each dataset and given $k$, the $n_{cand}$ parameter is then fine-tuned to maximize QPS at a target recall of 0.95, under the constraint that the configuration must also be capable of achieving 0.98 recall, which is listed in the full version due to page limitations~\cite{fullversion}. For {HNSW}, we increase $ef_{\text{search}}$ from $k$ in increments of $\frac{k}{2}$.  
For {\bbc}, \yin{since our CPU, like most modern CPUs, has an L1 cache of $C_{L1} = 32$ KB, we set the number of buckets $m$ according to Equation~\ref{eq:select-m}: 56 for Wiki, 
80 for 
C4, 
92 for 
{MSMARCO}, and 120 for {Deep100M}.}



\noindent\textbf{Implementations.} The baselines and our method 
are all implemented in C++. First, 
we use the hnswlib implementation~\cite{nmslib_hnswlib}, {a widely adopted industry-standard library}, for {HNSW}. 
For {IVF+RaBitQ}, we adopt its open-source implementation~\cite{gao2024rabitq}; and for {IVF} and {IVF+PQ}, we implement these methods based on {IVF+RaBitQ} because they share a common index structure. 
{Second}, for {Heap}
, we use the STL implementation. {For {d-Heap}, we use the Boost Library. For {Sorted}, we use its official implementation. For {Lazy}, we use the x86simdsort library. } 
All experiments are run on a machine equipped with an AMD Ryzen Threadripper PRO 5965WX 7.0GHz processor (
supporting 
AVX2) and 128 GB of RAM. 

\subsection{Experimental Results} \label{sec:qs-exp-results}

\begin{figure*}[!t]
\centering
\includegraphics[width=0.7\textwidth]{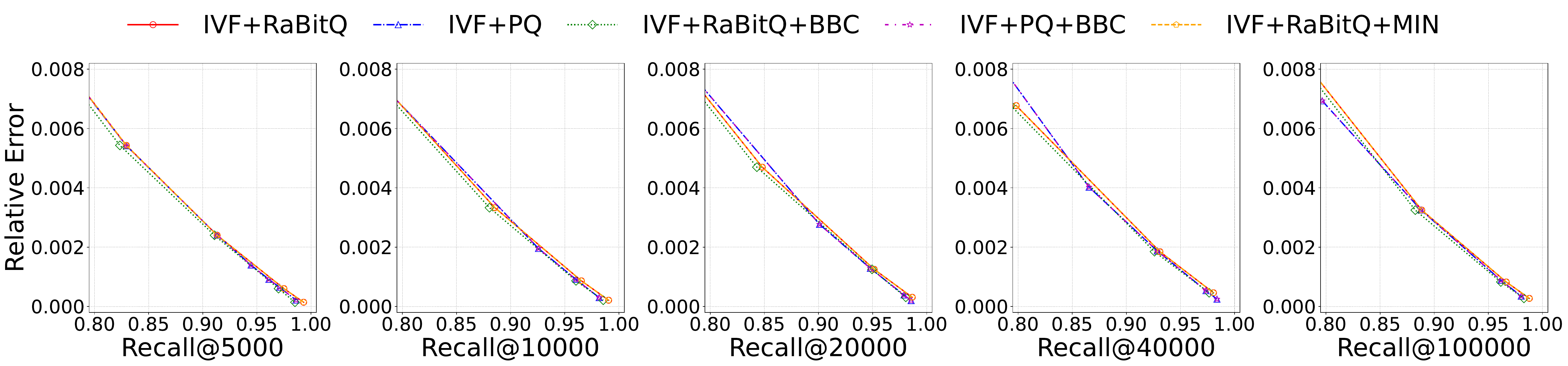}
\vspace*{-1em}
\caption{\revision{Relative error of large-$k$ ANN queries on the C4 dataset (lower is better).}}
\vspace*{-1em}
\label{fig:exp-error}
\end{figure*}

\begin{figure*}[!ht]
\centering
\includegraphics[width=0.7\textwidth]{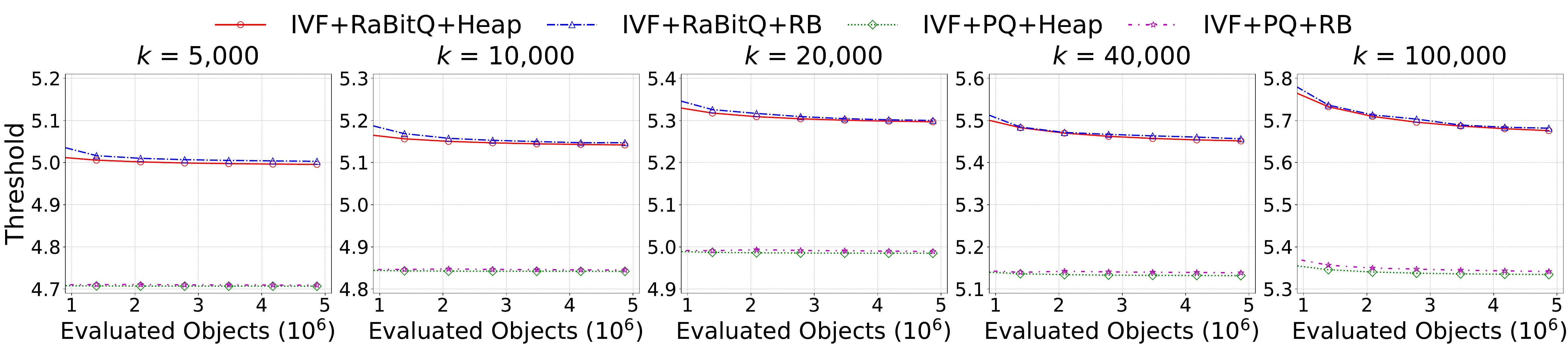}
\vspace*{-1em}
\caption{{Comparison of threshold values generated by the result buffer and binary heap on the {C4} Dataset.}}
\vspace*{-1em}
\label{fig:exp-threshold}
\end{figure*}

\begin{figure*}[!t]
\centering
\includegraphics[width=0.7\textwidth]{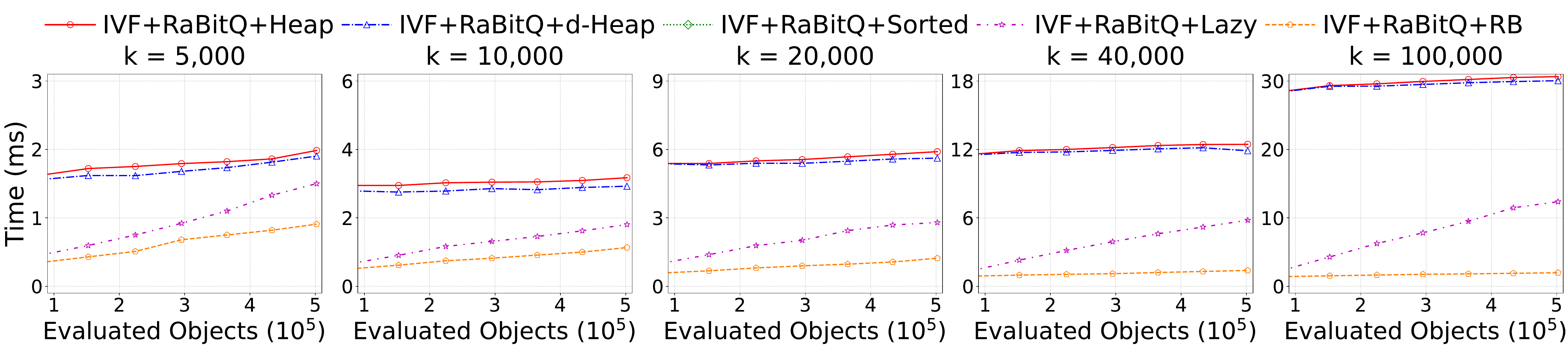}
\vspace*{-1em}
\caption{{Overhead of top-$k$ collectors under varying numbers of evaluated objects on the {Deep100M} dataset.}}
\vspace*{-1em}
\label{fig:exp-eff-topk-collector}
\end{figure*}

\noindent{\textbf{Exp-1: ANN Query Performance.}} 
The 
querying results over the {four} datasets are shown in Figure~\ref{fig:exp-acc-eff}. We have the following observations: \textbf{(1) 
{\bbc} achieves {1.4}$\times$--3.8$\times$ speedup over existing quantization-based methods for large-$k$ ANN queries.} 
Specifically, on the {Deep100M} dataset, when $k$ = 100,000 and recall@$k$ = 0.95, {IVF+PQ+{\bbc}} requires 61 ms per query, compared to 233 ms per query for {IVF+PQ}, achieving a 3.8$\times$ speedup. When $k = 5{,}000$, on the 
{{Wiki}} dataset at recall@$k$ = 0.98, {IVF+{PQ+\bbc}} achieves a {1.4}$\times$ acceleration over {IVF+{PQ}} (42 ms vs. 61 ms). 
The gain stems from the high efficiency of our proposed result buffer and the effectiveness of the newly designed re-ranking algorithm. 
\textbf{(2) The acceleration provided by {\bbc} becomes more significant as $k$ increases.} 
For example, on {Deep100M} at recall@$k$ = 0.95, {as $k$ increases from 5,000 to 100,000}, the acceleration ratio of {IVF+PQ+{\bbc}} over {IVF+PQ} increases from 2.9$\times$ at $k=5{,}000$ (2.8 ms vs. 8.2 ms) to 3.8$\times$ at $k=100{,}000$ (61 ms vs. 233 ms). This is because, as $k$ increases, existing collectors and re-ranking algorithms incur significantly higher costs, whereas our proposed result buffer remains efficient ({Exp-3}), and the newly designed ranking algorithm can further reduce the re-ranking cost substantially ({Exp-5}).
\textbf{(3) When $k$ is large, {IVF+RaBitQ+MIN} 
is significantly slower than {IVF+RaBitQ+{\bbc}}, and even slower than {IVF+RaBitQ}.} Across the four datasets, {IVF+RaBitQ+MIN} performs worse than 
{IVF+RaBitQ} and {IVF+RaBitQ+{\bbc}} at $k$ = 5,000 and 
the performance gap 
widens considerably as $k$ increases. This can be attributed to the high cost of heap operations, which outweighs the efficiency gained from re-ranking fewer objects and leads to higher overhead when $k$ is large, consistent with our previous analysis. \revision{\textbf{(4) {IVF+RaBitQ+{\bbc}} and {IVF+PQ+{\bbc}} are comparable to {IVF+RaBitQ} and {IVF+PQ} at $k = 100$, while outperforming them when $k \geq 500$, with the performance gap widening as $k$ increases.} This suggests {\bbc} 
should be used when $k \geq 500$ to improve query performance. \textbf{(5) Brute-force search is significantly slower than existing ANN indexes.} For example, on the Wiki dataset, at recall$@100,000 = 0.98$, BFC takes 1.8s per query, whereas the slowest ANN method, IVF+PQ, requires only 0.5s.}

\noindent\revision{\textbf{Exp-{2}: Relative error for large-$k$ ANN queries.} 
Figure~\ref{fig:exp-error} reports the relative error of large-$k$ ANN queries on the C4 dataset, and similar trends are observed across the other datasets. The results show that \textbf{(1) As recall increases, the relative error drops sharply from 0.3\% (recall@5,000 = 0.89) to 0.01\% (recall@5,000 = 0.98)}, and \textbf{(2) the relative error remains nearly constant across different $k$.} These  indicate that retrieved results remain highly accurate for large-$k$ ANN queries, regardless of the value of $k$.}

\noindent{\textbf{Exp-\revision{3}: Latency of Top-$k$ Collectors.}} We evaluate the time overhead of our proposed result buffer, denoted as {RB}, 
and compare it with {four} baselines: {Heap}, {{d-Heap}, {Sorted},} and {Lazy}. Using quantized distances from {IVF+RaBitQ} 
as keys, we use these methods to gather top-$k$ results, with $k$ ranging from 5,000 to 100,000, under varying numbers of evaluated objects. The number of evaluated objects is controlled by 
$n_{probe}$, which is tuned as in the previous experiment. The {isolated} time overhead of these methods is measured by VTune. Due to page limitations, we present the results on {Deep100M} in Figure~\ref{fig:exp-eff-topk-collector} {for {IVF+RaBitQ}}, while similar trends are observed on other datasets {and {IVF+PQ}}. 
We find that: \textbf{(1) Our result buffer {RB} is significantly faster than existing collectors, achieving up to an order of magnitude improvement.} For example, on the {Deep100M} dataset, when $k$ = 100{,}000 and $n_{probe}$ = 210, for {IVF+RaBitQ}, 
{RB} takes only {2.0} ms, compared to {30.6} ms for {Heap} and {12.3} ms for {Lazy}, achieving an order-of-magnitude speedup. \textbf{(2) {RB} remains efficient at higher values of $k$, while {Lazy} and {Heap} face significant degradation.} For example, on  {Deep100M}, when $k$ = 5{,}000 and $n_{probe}$ = 210, for {IVF+RaBitQ}, {RB} takes {0.9} ms, compared to {1.5} ms for {Lazy} and 1.9 ms for {Heap}. When $k$ = 100{,}000 and $n_{probe}$ = 210, {RB} takes {2.0} ms, while {Lazy} takes {12.3} ms and {Heap} takes {30.6} ms. This efficiency is due to {RB} avoiding maintaining an exact top-$k$ set throughout, which is consistent with previous experimental results. \revision{We further report the overall time for top-$k$ collection, including collectors, quantized distance computation, and other operations, in the technical report~\cite{fullversion} due to page limitations. We also include the total L1 cache misses during top-$k$ collection in the technical report, measured using perf profiling to validate that our speedup stems from reduced cache misses.}

\noindent{\textbf{Exp-\revision{4}: Gap between the relaxed and exact thresholds.}} We compare the threshold values generated by our proposed result buffer with those produced by the binary heap, thereby empirically validating the effectiveness of the equal-depth method. 
The experimental results on the {C4} dataset are shown in Figure~\ref{fig:exp-threshold}. \textbf{It indicates that the gap between the relaxed and exact threshold values 
remains consistently small, with differences on the order of $10^{-2}$, {thereby validating the effectiveness of Theorem~\ref{th:error}. 
}}

\begin{figure*}[!t]
\centering
\subcaptionbox{{Re-ranking Time}}{
\includegraphics[width=0.7\textwidth]{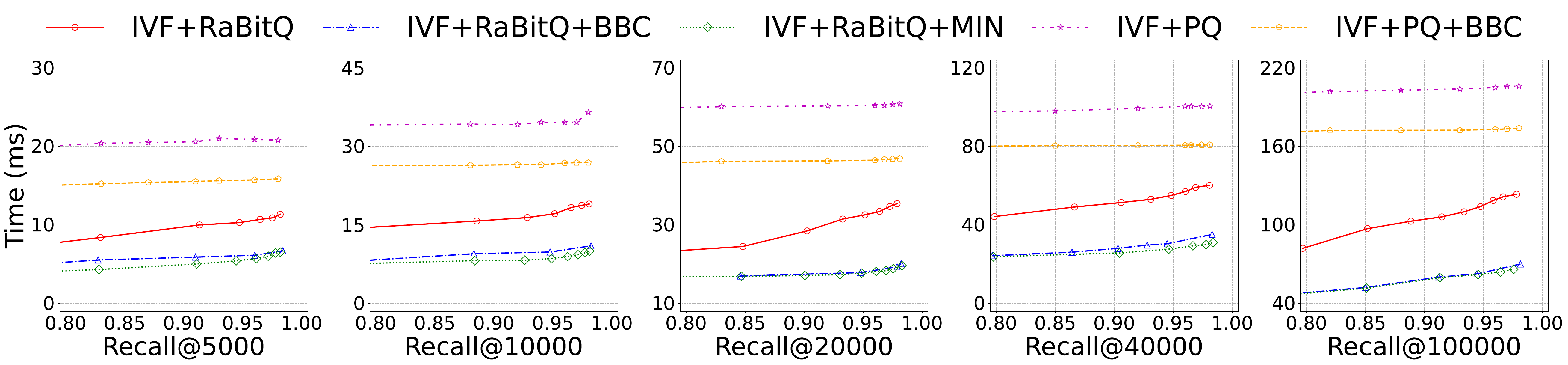}
}
\subcaptionbox{{Number of Re-ranked Objects}}{
\includegraphics[width=0.7\textwidth]{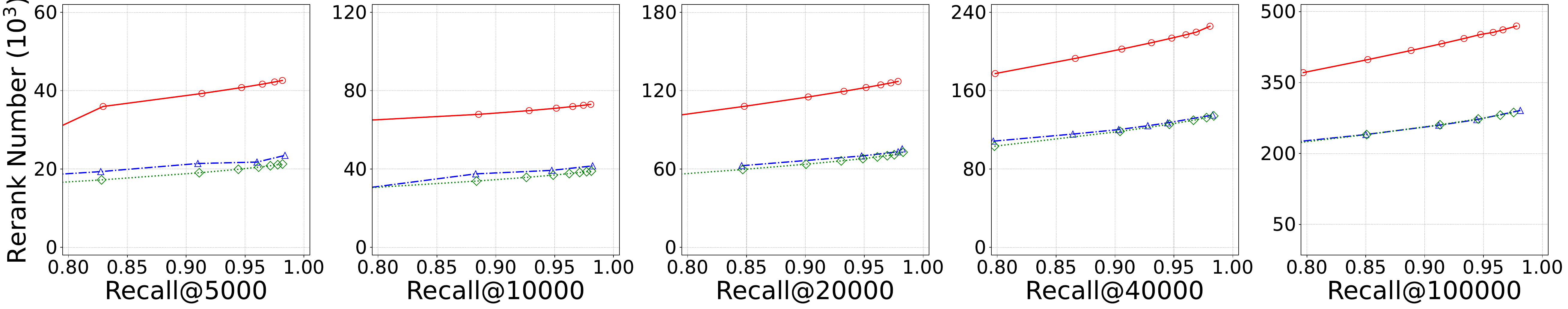}
}
\vspace*{-1em}
\caption{{Re-ranking Time and Number of Re-ranked Objects on the {C4} Dataset}}
\vspace*{-1em}
\label{fig:exp-eff-rerank}
\end{figure*}

\noindent{\textbf{Exp-\revision{5}: Comparison of Re-ranking Algorithms.}} We compare our proposed re-ranking algorithm with the naive algorithm, evaluating both their time overhead and the number of items re-ranked. In particular, Figure~\ref{fig:exp-eff-rerank} presents the re-ranking times and the number of re-ranked objects for {IVF+PQ}, {IVF+PQ+{\bbc}}, {IVF+RaBitQ}, and {IVF+RaBitQ+{\bbc}} on the {C4} dataset, based on VTune profiling results. The parameter $n_{probe}$ varies as in the previous experiments. 
We report recall@$k$ and the corresponding re-ranking times/re-ranked objects for $k$ = $5{,}000$, $10{,}000$, $20{,}000$,  $40{,}000$, and $100{,}000$. 
The key observations are as follows: \textbf{(1) Our re-rank algorithms significantly accelerate the re-ranking process.} Specifically, when $k$ = 20{,}000 and $\text{recall}@k$ = 0.95, {IVF+RaBitQ+{\bbc}} requires about 32 ms per query, whereas {IVF+PQ} takes around 18 ms, achieving a 1.8$\times$ speedup. Similarly, when $k$ = 20{,}000 and $\text{recall}@k$ = 0.95, {IVF+PQ+{\bbc}} requires approximately 45 ms per query for re-ranking, compared to about 60 ms for {IVF+PQ}, resulting in a 1.3$\times$ speedup. The speedup of {IVF+RaBitQ+{\bbc}} mainly results from a reduction in the number of re-ranked items, while the speedup of {IVF+PQ+{\bbc}} can be attributed to reduced cache misses, as detailed below. 
\textbf{(2) The number of re-ranked objects in {IVF+RaBitQ+{\bbc}} is significantly reduced compared to {IVF+RaBitQ}.} Specifically, when $k$ = 100{,}000 and $\text{recall}@k$ = 0.95, {IVF+RaBitQ} re-ranks 450,067 objects, whereas {IVF+RaBitQ+{\bbc}} re-ranks 223,142 objects, representing a reduction of nearly 50\%. This result is consistent with the 1.8$\times$ speedup observed above. Notably, the number of re-ranked objects in {IVF+RaBitQ+{\bbc}} is only slightly higher than that in the minimal re-ranking scenario of {IVF+RaBitQ+MIN}, demonstrating the effectiveness of our method. {\textbf{(3) The L1 cache miss count of {IVF+PQ+BBC} is significantly lower than that of {IVF+PQ}.} As shown in Table~\ref{tab:cache_2}, when $k = 100,000$, the L1 cache miss count is reduced from 7.61 $\times 10^5$ to 5.23 $\times 10^5$, corresponding to a 1.45$\times$ reduction, which demonstrates the effectiveness of Algorithm~\ref{alg:early-rerank-opq}.}

\begin{table}[t]
\centering
\small
\renewcommand{\arraystretch}{1.1}
\setlength{\arrayrulewidth}{0.8pt}
\caption{{L1 cache miss counts ($10^5$) of different re-ranking methods under varying values of $k$ when $n_{probe} = 500$.}}
\vspace*{-0.5em}
\label{tab:cache_2}
\begin{tabularx}{0.45\textwidth}{
    >{\centering\arraybackslash}p{1.4cm}
    *{5}{>{\centering\arraybackslash}X}
}
\toprule
${k}$ 
  & $5{,}000$ & $10{,}000$ & $20{,}000$ & $40{,}000$ & $100{,}000$ \\
\midrule
{IVF+PQ} & 1.03 & 1.74 & 3.23 & 5.42 & 7.61 \\
{IVF+PQ+BBC} &0.83 & 1.37 & 2.59 & 3.61 & 5.23 \\
\bottomrule
\end{tabularx}
\end{table}

\begin{table}[t]
\centering
\small
\renewcommand{\arraystretch}{1.1}
\setlength{\arrayrulewidth}{0.8pt}
\caption{{Top-$k$ collection time (ms) under varying values of $m$ on the {C4} dataset ($n_{probe}$ = 90)}}
\vspace*{-0.5em}
\label{tab:sensitivity}
\begin{tabularx}{0.45\textwidth}{
    >{\centering\arraybackslash}p{1cm}
    *{5}{>{\centering\arraybackslash}X}
}
\toprule
${m \backslash k}$ 
  & $5{,}000$ & $10{,}000$ & $20{,}000$ & $40{,}000$ & $100{,}000$ \\
\midrule
8 & 4.74 & 5.22 & 5.65& 5.91& 6.28\\
32 & 3.59 & 3.74& 3.99& 4.28 & 4.91 \\
80 & 3.50 & 3.68& 3.82& 4.14 & 4.81 \\
128 & 3.58 & 3.72 & 3.89 & 4.16 & 4.87 \\
256 & 3.70& 3.74 & 3.91 & 4.28 & 5.00 \\
\bottomrule
\end{tabularx}
\end{table}

\noindent{\textbf{Exp-\revision{6}: Parameter Sensitivity Study.} We evaluate the impact of the number of buckets $m$ following the procedure in Exp-3. Specifically, we measure the top-$k$ collection latency of {IVF+RaBitQ+BBC} on the {C4} dataset with $n_{probe} = 90$ and vary $m$ from 8 to 256. 
The experimental results are shown in Table~\ref{tab:sensitivity}, where $m=80$ is the optimal number of buckets as determined by Equation~\ref{eq:select-m}. The results indicate that (1) the value of $m$ computed by Equation~\ref{eq:select-m} achieves the lowest latency; (2) small deviations from this value result in only marginal latency increases; (3) very small values of $m$ cause a substantial latency increase because objects become concentrated in a few buckets, making the final selection costly; and (4) very large values of $m$ also increase latency due to more frequent L1 cache misses induced by the large $m$
.}

%% file: related.tex
\section{{Related Work}}
\label{sec:related}


\noindent\textbf{Approximate Nearest Neighbor Search.}  
%
Various Approximate Nearest Neighbor (ANN) search methods
have been proposed~\cite{patella2009approximate, pan2024vector, ite_matsui_2018,liApproximateNearestNeighbor2020}, which are typically classified into four categories: graph-based methods~\cite{fu2019fast,DBLP:journals/pami/FuWC22,DBLP:journals/is/MalkovPLK14,malkovEfficientRobustApproximate2020,harwood2016fanng,wang2021comprehensive,azizi2025graph, yin2026distvs,yin2026gps,gui2026pilotann}, quantization methods~\cite{DBLP:journals/pami/GeHK014,gong2013iterative,jegouProductQuantizationNearest2011,DBLP:conf/eccv/MartinezZHL18,DBLP:conf/icde/PaparrizosELEF22,DBLP:conf/mm/TuncelFR02,DBLP:conf/icml/ZhangDW14,gao2024rabitq, rabitq2, blink2024,DBLP:conf/cvpr/BabenkoL14,DBLP:conf/vldb/WeberSB98,DBLP:conf/sc/JiangLZLHSRZRHA23,ferhatosmanoglu2000vector}, hashing-based methods~\cite{nagarkar2018pslsh,LSHDatarIIM04, wangSurveyLearningHash2018, DBLP:journals/pami/HeoLHCY15, DBLP:journals/tods/TaoYSK10, DBLP:journals/tkde/TianZZ24, DBLP:journals/pvldb/LuWWK20, DBLP:journals/pvldb/HuangFZFN15,DBLP:conf/mm/TuncelFR02,DBLP:conf/sigmod/GanFFN12,DBLP:conf/sigmod/LeiHKT20, DBLP:conf/sigmod/LiYZXCLNC18,pham2022falconn}, and tree-based methods~\cite{yianilos1993data,arora2018hdindex, BeygelzimerKL06Covertree, RamS19kdtree,ciaccia1997m}. 
Among 
these methods, IVF and graph-based indexes are widely used in industry 
~\cite{pan2024vector,DBLP:conf/sigmod/LiZAH20} and quantization methods have proven highly effective in {saving memory and} accelerating {query processing}~\cite{yahoojapan2018ngt,DBLP:conf/icml/GuoSLGSCK20,DBLP:conf/mir/AndreKS17,gao2024rabitq}. 
For a comprehensive overview, we refer readers to recent tutorials~\cite{DBLP:journals/pvldb/EchihabiPZ21}, benchmark/experimental evaluations~\cite{simhadri2022results, DBLP:journals/is/AumullerBF20, DBLP:journals/debu/0001C23, dobson2023scaling, DBLP:journals/debu/00070P023}, and surveys~\cite{pan2024vector, DBLP:journals/vldb/PanWL24, ite_matsui_2018} for details. 
Although ANN queries have been extensively studied, to the best of our knowledge, the large-$k$ ANN query studied in this work has not yet been specifically investigated. \remove{In addition, Large-$k$ ANN queries 
cannot be effectively solved with range queries. 
In practice, range queries in vector search remain underexplored and technically difficult to utilize
for two reasons: (1) lack of intuition and semantic meaning of similarity ranges. In high-dimensional spaces, the semantic meaning of a similarity threshold (e.g., 0.9) is opaque: the same threshold may return few results for some queries, indicating high similarity, but hundreds of thousands objects for other queries, indicating lower similarity, making it difficult to specify an effective threshold; and 
(2) high uncertainty in result cardinality: due to the distance concentration phenomenon, the distances from a query vector to data vectors often lie within a very narrow range (e.g., 0.9–1.0), so even a small change in the threshold can result in a dramatic change in the number of returned objects, which makes the results challenging to use effectively.}






\noindent\textbf{Quantization.} We focus on improving quantization based methods. The quantization of high-dimensional vectors has been extensively explored in the literature~\cite{matsui2018reconfigurable,matsui2015pqtable,DBLP:conf/icml/GuoSLGSCK20,DBLP:journals/pami/GeHK014,gong2013iterative,jegouProductQuantizationNearest2011,DBLP:conf/icde/PaparrizosELEF22,DBLP:conf/mm/TuncelFR02,DBLP:conf/icml/ZhangDW14,ite_matsui_2018,gao2024rabitq, rabitq2, blink2024,DBLP:conf/cvpr/BabenkoL14,DBLP:conf/vldb/WeberSB98,DBLP:conf/sc/JiangLZLHSRZRHA23}. {
Early research on quantization focuses on reducing quantization error, with Product Quantization (PQ)~\cite{jegouProductQuantizationNearest2011,DBLP:journals/pami/GeHK014} as the representative 
method. 
Recently, 
RaBitQ~\cite{gao2024rabitq,rabitq2} is proposed, which provides theoretical 
bounds for estimated distances. 
With the help of SIMD-based implementations (a.k.a. FastScan), these methods have achieved great success in accelerating other ANN approaches~\cite{DBLP:journals/pvldb/AndreKS15,yahoojapan2018ngt,DBLP:conf/icml/GuoSLGSCK20,gao2024rabitq,rabitq2,DBLP:journals/pacmmod/GouGXL25}. In this study, we integrate our proposed {\bbc} into \textsf{IVF+PQ} and {IVF+RaBitQ} to demonstrate its plug-and-play ability to improve the efficiency of quantization methods for 
large-$k$ ANN queries. {There are methods that integrate quantization and graph-based approaches. However, these early explorations exhibit limited scalability. For instance, SymphonyQG~\cite{DBLP:journals/pacmmod/GouGXL25} and NGT-QG~\cite{yahoojapan2018ngt} encounter out-of-memory errors during indexing on our system due to their substantial memory consumption, which exceeds the available capacity. This is consistent with prior experimental findings~\cite{DBLP:journals/pacmmod/GouGXL25}. Moreover, LVQ~\cite{DBLP:journals/pvldb/AguerrebereBHTW23} is closed-source. Therefore, we leave the integration of \yin{\bbc} with these methods for future work.}

\noindent\textbf{Priority Queue.} Priority queues have been extensively studied, with the binary heap and its variants being the most common implementations due to their $O(\log(n))$ insertion and deletion time complexity~\cite{williams1964algorithm}. However, in large-$k$ ANN queries, this theoretical efficiency is no longer effective, as L1 cache misses 
dominate 
the runtime overhead. This observation is consistent with previous experimental evaluations~\cite{larkin2014back}, which report a strong correlation between the priority queue's processing time and L1 cache miss rates. Consequently, studies that focus on optimizing the theoretical time complexity of priority queues~\cite{pugh1990skip,vuillemin1978data,fredman1987fibonacci}, such as the Fibonacci heap, offer limited benefits in our context. It is noted that \cite{DBLP:journals/pvldb/AguerrebereBHTW23} replaces the heap with a sorted linear buffer.  
Each insertion locates the proper position and shifts all subsequent elements backward to maintain order. The linear buffer's sequential layout facilitates hardware prefetching, thereby reducing L1 cache misses and outperforming the heap when $k$ is small~\cite{DBLP:journals/pvldb/AguerrebereBHTW23}. However, its advantage vanishes with larger $k$ due to the $O(k)$ insertion cost. Several top-k collectors are designed for GPUs~\cite{sismanis2012parallel,tang2015efficient,johnsonBillionscaleSimilaritySearch2019}, such as FAISS's WarpSelect. However, as reported in~\cite{johnsonBillionscaleSimilaritySearch2019}, these 
methods also face performance degradation when $k$ is large.



%% file: conclusion.tex
\section{Conclusions and Future Directions}

In this paper, we propose a novel bucket-based result collector (\yin{\bbc}) to accelerate quantization-based methods for large-$k$ ANN queries, which consists of two key components: a bucket-based result buffer and two re-ranking algorithms. One potential future direction is to modify the \yin{\bbc} to support graph-quantization methods for large-$k$ ANN queries.\remove{Graph-based methods typically employ greedy beam search to retrieve top-$k$ results. Starting from an entry point, the search iteratively explores the neighbors of the currently closest node, where a min-heap is used to dynamically maintain the nearest candidates. When $k$ is large, the min-heap inevitably suffers from frequent L1 cache misses and increased latency. Our proposed \bbc can address this by maintaining only a small heap in the nearest bucket.}
\yin{Another promising direction is adapting \bbc for GPU settings to accelerate batch large-$k$ ANN queries. 
} 

%% file: appendix.tex
\newpage
\appendix

\section{The Correctness Proof of Algorithm~\ref{alg:RabitQ-optimal}}

\begin{proof}
\label{pf:rabitq-optimal}
As described in Algorithm~\ref{alg:RabitQ-optimal}, after the scanning phase (lines 2-11), we obtain two heaps: a max-heap ${H}_u$ of size $k$ and a min-heap ${H}_l$. Clearly, the top-$k$ results returned by the \textsf{RaBitQ} must be contained in ${H}_u$ or ${H}_l$. There are two possible cases:
\begin{enumerate}
    \item \textbf{${H}_u$ are the final top-$k$ results.} In this scenario, the object at the top of ${H}_u$ cannot be avoided for re-ranking. This is because its upper bound is the $k$-th largest and thus greater than $\mathrm{Dist}_k$ (the exact distance of the $k$-th object), while its lower bound must be less than $\mathrm{Dist}_k$ since it is included in the top-$k$ results.

    \item \textbf{Some objects in ${H}_l$ may enter the final top-$k$.} In this case, the object at the top of ${H}_l$ must also be re-ranked. This is because its lower bound is smaller than those of other objects in ${H}_l$ that enter the top-$k$, meaning its lower bound is less than $\mathrm{Dist}_k$. At the same time, its upper bound exceeds the $k$-th upper bound, and thus is also greater than $\mathrm{Dist}_k$.
\end{enumerate}
Therefore, in both cases, at least one of the objects at the top of ${H}_u$ or ${H}_l$ must be re-ranked. We further prove that the object with the smaller lower bound between the two is guaranteed to require re-ranking. Similarly, there are two cases to consider. 

\begin{enumerate}
    \item \textbf{If the object at the top of ${H}_u$ has the smaller lower bound, it must be re-ranked.} If it belongs to the top-$k$, it must be re-ranked as discussed above. If it does not belong to the top-$k$, its lower bound remains smaller than that of the top object in ${H}_l$, and its upper bound is the $k$-th upper bound; therefore, it must also be re-ranked.
    \item \textbf{If the object at the top of ${H}_l$ has the smaller lower bound, it is also unavoidable for re-ranking.} This is because its predicted distance interval entirely covers that of the object at the top of ${H}_u$, and thus intersects with $\mathrm{Dist}_k$.
\end{enumerate}
The above consideration is based on the case where neither of the top objects from the two heaps has been re-ranked. We now consider the case where the top object in $H_u$ has already been evaluated.

\begin{enumerate}
    \item \textbf{${H}_u$ are the final top-$k$ results.} In this scenario, the distance between the object at the top of $H_u$ and the query is $\mathrm{Dist}_k$. If the lower bound of the top of ${H}_l$ is smaller than $\mathrm{Dist}_k$, it needs to be re-ranked.

    \item \textbf{Some objects in \( H_l \) may enter the final top-\( k \).} In this case, the object at the top of \( H_l \) must also be re-ranked, as in the previous situation.
\end{enumerate}

\end{proof}

\begin{figure*}[!t]
\centering
\includegraphics[width=0.9\textwidth]{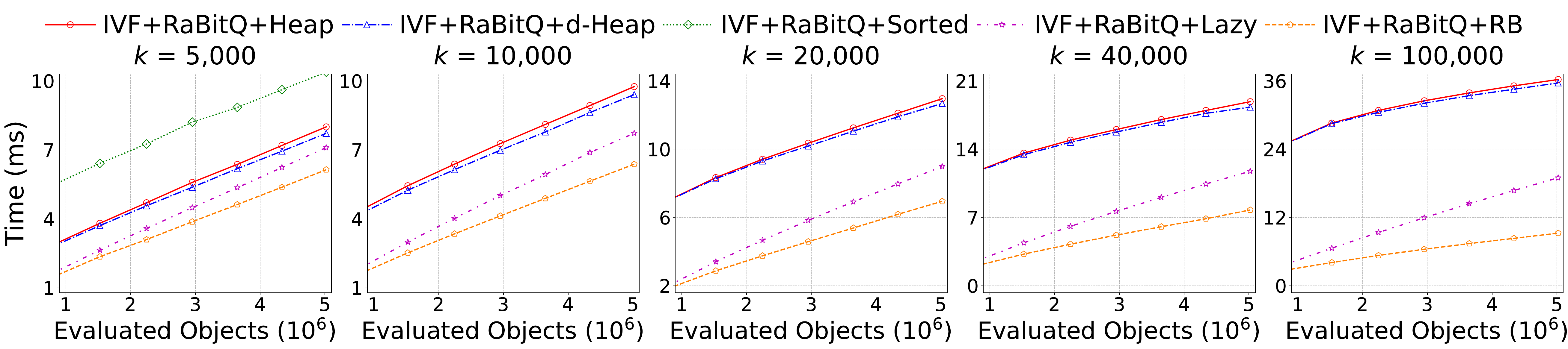}
\caption{{Top-$k$ collection time under different numbers of evaluated objects on the {Deep100M} dataset.}}
\label{fig:exp-eff-topk}
\end{figure*}


\section{Proof of Theorem~\ref{th:error}}
\begin{table*}[t]
    \centering
    \caption{{The $n_{\text{cand}}$ parameter of IVF+PQ across different datasets and $k$ settings.}}
    \vspace*{-1em}
    \label{tab:numcand}
    \renewcommand{\arraystretch}{1.2}
    \setlength{\arrayrulewidth}{0.8pt}
    \begin{tabularx}{0.8\textwidth}{
         *{5}{>{\centering\arraybackslash}X}
    }
        \toprule
        \textbf{$k$} & {WiKi} & {C4} & {MSMARCO} & {Deep100M} \\
        \midrule
        10 & 500 & 500 & 600 & 1,000 \\
        100 & 2,500 & 2,500 & 3,000 & 6,000 \\
        500 & 8,000 & 8,000  & 20,000 & 20,000 \\
        1,000 & 15,000 & 15,000 & 40,000 & 40,000 \\
        2,500 &30,000 & 30,000 & 50,000 & 60,000 \\
        5,000 & 50,000  &50,000 & 70,000 & 100,000 \\
        10,000 & 80,000 & 80,000 & 100,000 & 140,000 \\
        20,000 & 150,000 &140,000 & 180,000 & 240,000 \\
        40,000 & 240,000 &240,000 & 280,000 & 320,000 \\
        100,000 & 500,000&500,000 & 600,000 & 700,000 \\
        \bottomrule
    \end{tabularx}
\end{table*}


\begin{proof}
We define the quantized variable over the interval $[0, \mu]$ by
\[
\widehat R = \sum_{i=1}^m b_{i+1} \,\mathbf 1\{R \in (b_{i},b_{i+1}]\}.
\]
By this definition, whenever $R \in (b_i, b_{i+1}]$, $\widehat R$ takes the upper boundary $b_{i+1}$, which guarantees that $R \le \widehat R$ within this domain.

\medskip
\noindent\textbf{Step 1. Integral representation of the error.}
The expected absolute quantization error over $[0, \mu]$ can be expressed as the area between the cumulative distribution functions (CDFs) of $R$ and $\widehat R$. 
Letting $\Delta(x) = \mathbb{P}(R \le x) - \mathbb{P}(\widehat R \le x)$ denote the CDF difference, we obtain:
\[
\mathbb E|R-\widehat R| = \int_{0}^{\mu} \Delta(x) \,dx.
\]
Note that $\mathbb{P}(\widehat{R} \le x) \le \mathbb{P}(R \le x)$, leading to $\Delta(x) \ge 0$ for all $x$. Moreover, $\Delta(x) \le \mathbb{P}(R \le x) \le 1$.

\medskip
\noindent\textbf{Step 2. Partitioning the integration region.}
To bound the integral, we partition the domain $[0, \mu]$ into two distinct regions: $[0, z]$ and $[z, \mu]$. We choose the transition point $z$ according to the sub-Gaussian concentration of the left tail:
\[
z = \mu - \sqrt{\frac{\log(2 c_1 m)}{c_0 d}}.
\]
The total quantization error is simply the sum of the integral contributions from these two regions:
\[
\int_{0}^{\mu} \Delta(x)\,dx = \int_{0}^{z} \Delta(x)\,dx + \int_{z}^{\mu} \Delta(x)\,dx.
\]


\medskip
\noindent\textbf{Step 3. Bounding the error on $[0, z]$.}
For the first interval $[0, z]$, we drop the negative term in $\Delta(x)$ to use the upper bound $\Delta(x) \le \mathbb{P}(R \le x)$. By the given sub-Gaussian left tail assumption, for any $x \le \mu$, we have $\mathbb{P}(R \le x) \le 2c_1 \exp(-c_0 d (\mu-x)^2)$. Integrating this over $[0, z]$ gives:
\[
\int_{0}^{z} \Delta(x) \,dx \le \int_{0}^{z} 2c_1 \exp\bigl(-c_0 d (\mu-x)^2\bigr) \,dx.
\]
We can further upper-bound this by extending the integration domain to $(-\infty, \mu]$:
\[
\int_{0}^{z} \Delta(x) \,dx \le \int_{-\infty}^{\mu} 2c_1 \exp\bigl(-c_0 d (\mu-x)^2\bigr) \,dx.
\]
To evaluate this, we apply the change of variables $t = \mu - x$, which implies $dx = -dt$. The integration limits change from $x \to -\infty$ to $t \to \infty$, and from $x = \mu$ to $t = 0$. Reversing the limits to absorb the negative sign, we have:
\[
\int_{-\infty}^{\mu} 2c_1 \exp\bigl(-c_0 d (\mu-x)^2\bigr) \,dx = \int_{\infty}^{0} 2c_1 \exp(-c_0 d t^2) (-dt) = \int_{0}^{\infty} 2c_1 \exp(-c_0 d t^2) \,dt.
\]
Evaluating this standard Gaussian integral (using the identity $\int_{0}^{\infty} e^{-a t^2}dt = \frac{1}{2}\sqrt{\frac{\pi}{a}}$ with $a = c_0 d$), we obtain:
\[
\int_{0}^{\infty} 2c_1 \exp(-c_0 d t^2) \,dt = 2c_1 \left( \frac{1}{2} \sqrt{\frac{\pi}{c_0 d}} \right) = c_1 \sqrt{\frac{\pi}{c_0 d}}.
\]
Thus, the contribution of the first region is strictly bounded by $c_1 \sqrt{\frac{\pi}{c_0 d}}$.

\medskip
\noindent\textbf{Step 4. Bounding the error on $[z, \mu]$.}
For the second region $[z, \mu]$, we apply the worst-case global bound $\Delta(x) \le 1$. The integral over this region is constrained simply by the geometric width of the interval:
\[
\int_{z}^{\mu} \Delta(x) \,dx \le \int_{z}^{\mu} 1 \,dx = \mu - z.
\]
Substituting our predefined choice of $z$, this contribution evaluates exactly to:
\[
\mu - z = \sqrt{\frac{\log(2 c_1 m)}{c_0 d}}.
\]

\medskip
\noindent\textbf{Conclusion.}
Summing the individual bounds derived in Step 3 and Step 4 yields the total expected quantization error over $[0, \mu]$:
\[
\mathbb E|R-\widehat R| \le c_1 \sqrt{\frac{\pi}{c_0 d}} + \sqrt{\frac{\log(2 c_1 m)}{c_0 d}}.
\]
This completes the proof.
\end{proof}

\section{Experimental Settings}
{Details of the evaluated methods are provided below:}
\begin{itemize}[leftmargin=*, topsep=0pt]
    \item {IVF~\cite{jegouProductQuantizationNearest2011}}: A representative ANN index. The accuracy-efficiency trade-off is controlled by the 
    hyperparameter $n_{probe}$. 
    \item {HNSW~\cite{malkovEfficientRobustApproximate2020}}: A popular graph-based ANN index. The accuracy-efficiency trade-off is controlled by the 
    hyperparameter $ef_{search}$.    
    \item {IVF+PQ}~\cite{DBLP:journals/pami/GeHK014}: This method integrates the representative unbounded quantization technique, product quantization ({PQ}), with the IVF index.   At query time, each query is routed to the $n_{probe}$ nearest clusters, within which the search procedure of {PQ} is applied, as detailed above. The accuracy-efficiency trade-off is controlled by the hyperparameters $n_{cand}$ and $n_{probe}$. 
    \item {IVF+RaBitQ}~\cite{gao2024rabitq}: This method integrates the representative bounded quantization method {RaBitQ} with the IVF index. At query time, each query is routed to the $n_{probe}$ nearest clusters, within which the search procedure of {RaBitQ} is applied, as detailed above. The accuracy-efficiency trade-off is controlled by the hyperparameter $n_{probe}$. 
\end{itemize}
The dataset statistics are shown below.
\begin{itemize}[leftmargin=*, topsep=0pt]

\item {\textsf{WiKi}: The \textsf{WiKi} dataset comprises 10 million corpus sampled from the Wikipedia dataset\footnote{\url{https://huggingface.co/datasets/Cohere/wikipedia-22-12-en-embeddings}}, which serves as the open-source backbone for large language model pre-training. The embeddings are generated using the gte-Qwen2-1.5B-instruct model~\cite{qwen3embedding}. The query set~${Q}$ consists of 1,000 randomly selected passages.}


\item \textsf{C4}: The \textsf{C4} dataset\footnote{\url{https://huggingface.co/datasets/allenai/c4/}} is a large-scale 
open-source corpus designed for natural language pre-training. 
We select 40 JSON files from the training set, which 
contains over 14 million passages, and generate embeddings using the T5 model~\cite{2020t5}. 
The query set~${Q}$ consists of 1,000 randomly selected passages. 



\item \textsf{MSMARCO}: The \textsf{MSMARCO} dataset\footnote{\url{https://huggingface.co/datasets/Snowflake/msmarco-v2.1-snowflake-arctic-embed-m-v1.5}} comprises 18 million passages sampled from the MSMARCO-V2.1 dataset, which is used as the corpora for The TREC 2024 RAG Track
. The embeddings are generated using the Snowflake's Arctic-embed-m-v1.5 model~\cite{snowflake2024embedding}. The query set~${Q}$ consists of 1,000 randomly selected passages.

\item \textsf{Deep100M}: The \textsf{Deep100M} dataset~\cite{simhadri2022results} is the largest benchmark commonly used for ANN evaluation. The embeddings are obtained from an image descriptor dataset, where each embedding is produced by 
the GoogLeNet model~\cite{szegedy2015going}. 
The official 100K query set is used in our experiments.


\end{itemize}
Table~\ref{tab:numcand} reports the optimal $n_{cand}$ values for different datasets and $k$. For each dataset–$k$ combination, $n_{cand}$ is determined to maximize QPS at a recall of 0.95, while ensuring that recall can reach 0.98.

\section{Latency of Top-$k$ Collectors}

Figure~\ref{fig:exp-eff-topk} shows the results on {Deep100M} and similar trends are observed on other datasets. 
The results show that: \noindent\textbf{{{RB}} significantly accelerates top-$k$ collection based on estimated distances compared to {Heap} and {Lazy}.} In particular, on the {Deep100M} dataset, when $k$ = 100{,}000 and $n_{probe}$ = 210, for {IVF+RaBitQ}, {{RB}} buffer requires only {9.2} ms, compared to 36.2 ms for {Heap} and {18.9} ms for {Lazy}, achieving a {2.1}$\times$ speedup.

{We also include the total L1 cache misses during top-$k$ collection in  Table~\ref{tab:cache_1}, measured using Perf profiling. Since other components of the pipeline, such as estimated distance computation, are identical across methods, the observed differences in L1 cache misses mainly stem from the use of different collectors. The experimental results confirm that the acceleration of the top-$k$ collection process is primarily driven by reduced L1 cache misses, with the speedup closely tracking the reduction in miss counts. For example, when $k$ = 100,000, {RB} halves the L1 cache misses compared to {Lazy}, resulting in an 2$\times$ speedup.}

\begin{table}[htbp]
\centering
\renewcommand{\arraystretch}{1.1}
\setlength{\arrayrulewidth}{0.8pt}
\caption{{L1 cache miss counts ($10^5$) during top-$k$ collection process with different collectors under $k$ when $n_{probe} = 210$.}}
\label{tab:cache_1}
\begin{tabularx}{0.45\textwidth}{
    >{\centering\arraybackslash}p{1cm}
    *{5}{>{\centering\arraybackslash}X}
}
\toprule
${k}$ 
  & $5{,}000$ & $10{,}000$ & $20{,}000$ & $40{,}000$ & $100{,}000$ \\
\midrule
{Heap} & 3.2 & 3.7 & 4.8 & 7.2 & 9.7 \\
{d-Heap} &3.1 & 3.5 & 4.3 & 6.6 & 9.3 \\
{Sorted} & 4.5 & 8.9 & 19.9 & 69.3 & 396 \\
{Lazy} &2.9&  3.4 & 4.0 & 5.1 & 8.1 \\
{RB} & 2.7 & 2.9 & 3.1 & 3.4 & 4.1 \\
\bottomrule
\end{tabularx}
\end{table}

\section{Memory Cost}
Table~\ref{tab:qs_memory_cost} reports the memory cost of {{\bbc}} under different $k$ and $m$ settings. The results indicate that the memory cost of {\bbc} is negligible compared to the dataset sizes reported in Table~\ref{tab:qs-dataset} 
since {\bbc} introduces nearly no additional memory usage.

\begin{table}[htbp]
\centering
\renewcommand{\arraystretch}{1.2}
\setlength{\arrayrulewidth}{0.8pt}
\caption{{Memory cost (MB) of \bbc under different $m$ and $k$.}}
\label{tab:qs_memory_cost}
\begin{tabularx}{0.45\textwidth}{
    >{\centering\arraybackslash}p{1.2cm}
    *{5}{>{\centering\arraybackslash}X}
}
\toprule
${m \backslash k}$ 
  & $5{,}000$ & $10{,}000$ & $20{,}000$ & $40{,}000$ & $100{,}000$ \\
\midrule
56 & 2.1 & 4.3 & 8.5 & 17.1 & 34.2 \\ 
120 & 4.6 & 9.2 & 18.3 & 36.6 & 73.2 \\ 
\bottomrule
\end{tabularx}
\end{table}